\documentclass{article}
\usepackage[margin=1in]{geometry}

\usepackage[utf8]{inputenc}
\usepackage[T1]{fontenc}
\usepackage{xcolor}
\usepackage{graphicx}
\usepackage{svg}
\usepackage{subcaption}
\usepackage[font=small,labelfont=bf]{caption}
\usepackage{wrapfig}

\usepackage{amsmath,amssymb,amsthm,amsfonts}
\usepackage{mathtools}
\usepackage{bbm}

\usepackage{microtype}
\usepackage{nicefrac}
\usepackage{booktabs}
\usepackage{enumitem}

\usepackage{hyperref}
\hypersetup{
    colorlinks=true,
    linkcolor=blue,
    filecolor=magenta,
    urlcolor=red,
    citecolor=red
}
\usepackage[capitalise,nameinlink]{cleveref}

\usepackage{thmtools}

\usepackage{tcolorbox}
\usepackage{framed}
\usepackage{tikz}
\usepackage{xspace}
\usepackage{cite}
\usepackage{authblk}
\usepackage{mycommands}
\usepackage{algorithm}
\usepackage{algorithmic}
\usepackage{setspace}

\newtheorem{theorem}{Theorem}
\newtheorem{lemma}[theorem]{Lemma}
\newtheorem{proposition}[theorem]{Proposition}
\newtheorem{corollary}[theorem]{Corollary}
\newtheorem{observation}[theorem]{Observation}
\newtheorem{definition}[theorem]{Definition}

\newtheorem*{framework*}{Framework}

\title{A Framework for the Design of Efficient Diversification Algorithms \\for Packing and Covering Problems}

\author[1]{Waldo Gálvez}
\author[2]{Mayank Goswami}
\author[3]{Arturo Merino}
\author[4]{\\GiBeom Park}
\author[5]{Meng-Tsung Tsai}
\author[6]{Victor Verdugo}

\affil[1]{Universidad de Concepción, Chile, \texttt{wgalvez@inf.udec.cl}}
\affil[2]{Queens College of CUNY, USA, \texttt{mayank.goswami@qc.cuny.edu}}
\affil[3]{University of Chile, Chile, \texttt{amerino@dcc.uchile.cl}}
\affil[4]{CUNY Graduate Center, USA, \texttt{gpark1@gradcenter.cuny.edu}}
\affil[5]{Academia Sinica, Taiwan, \texttt{mttsai@iis.sinica.edu.tw}}
\affil[6]{Pontificia Universidad Católica de Chile, \texttt{victor.verdugo@uc.cl}}

\date{}

\begin{document}

\clearpage
\pagenumbering{arabic}

\maketitle
\begin{abstract}
    There has been considerable recent interest in computing a diverse collection of solutions to a given optimization problem.
Given a classical optimization problem $\Pi$ (e.g., minimum spanning tree, maximum independent set) with input size $n$ and an integer $k\geq 1$, the goal is to generate a collection of $k$ maximally diverse solutions to $\Pi$. This diverse-X paradigm not only allows the user to generate very different solutions, but also helps make systems more secure and robust by handling uncertainty and achieving energy efficiency. 

For several problems $\Pi$ solvable in polynomial time, there are efficient $\text{poly}(n,k)$-time exact or approximation algorithms available for their diverse-solution variants.
In contrast, only FPT algorithms are known for NP-complete problems such as vertex covers and independent sets.
In the worst case, these algorithms run in time $\exp((kn)^a)$ for some $a>0$.

In this work, we address this gap and give a general framework for $\text{poly}(n,k)$-time \emph{bi-criteria} PTAS for the diverse variants of several NP-hard problems such as \textsc{Knapsack}, \textsc{Maximum Weight Independent Set} (MWIS) and \textsc{Minimum Weight Vertex Cover} (MWVC) in planar graphs.
More precisely, given a maximization (or, minimization) problem $\Pi$, an integer $k \geq 1$, user-specified quality threshold $ c \in (0, 1) $, and user-controlled error parameters $\delta,\varepsilon \in (0,1)$, our bi-approximation algorithm returns a collection of $k$ solutions $S_1,S_2,\ldots,S_k$ that are $(1-\delta)c$-optimal and whose diversity is at least a $(1-\varepsilon)$-fraction of the maximum diversity achievable by any collection of $k$ $c$-optimal solutions.

To complement our positive results, we prove that the diverse-solution variants of \KP and MWIS mentioned above do not admit standard single-criterion PTASs.
In particular, we show that under the assumption $\PneqNP$, \emph{both} $c$, the solution-quality threshold, and the diversity guarantee must be approximated to achieve a polynomial time algorithm.
Moreover, even under such a simultaneous relaxation, none of the two problems admits a bi-criteria FPTAS.
\end{abstract}

\clearpage
\tableofcontents

\clearpage

\section{Introduction}

Computing a collection of diverse solutions to a given optimization problem has gained a lot of attention recently~\cite{fomin2020diverse,baste2019fpt,hanaka2021finding,hanaka2022computing,hanaka2023framework,abboud2022improved,gao2022obtaining,de2023finding,schulhof2025finding,misra2025parameterized,austrin2025algorithms,galvez2025computing,chakraborty2026maximizing,drabik2026finding,dehaleesan2026fpt,golovach2026polynomial}.
While classical algorithms are tailored to produce one solution, the task here is to output a collection of $k \geq 1  $ solutions that are \emph{maximally dispersed} in the solution space.
In general, one is given a diversity measure on the space of $k$-tuples of solutions to a problem, and the goal is to output a set of $k$ solutions that maximize this measure.

When solutions can be represented as subsets of the input, a natural metric for diversity measure is the size of the symmetric difference between two sets. More precisely, for a given optimization problem, let $\cF$ denote the family of subsets of the input that satisfy the problem's constraints—that is, the set of \emph{feasible solutions}. Then, given two feasible solutions $X \in \cF$ and $Y \in \cF$ of a given optimization problem, $d(X,Y)=|X \bigtriangleup Y|$. This is extended to a $k$-tuple of solutions by considering either the average, or the minimum pairwise distance between all pairs of solutions. For example, if $\cF$ is the family of all spanning trees of a given graph $G$, then the task is to find $k$ spanning trees whose average (or, minimum) of pairwise distances is maximized. We will mostly focus on the average for now, and discuss the minimum later.

The weighted setting is also of interest, e.g., $\cF$ is the family of all \emph{minimum} spanning trees of $G$. However, if $G$ has a unique minimum spanning tree, then the problem of returning $k$ diverse minimum spanning trees becomes uninteresting. The natural approach here is to enlarge the set of solutions by allowing approximations. We call such approximately optimal solutions {\it nice}. 
In fact, if we replace the minimum spanning tree above by an NP-hard problem, say, maximum weight independent set (MWIS), and we want polynomial time algorithms, then \emph{allowing approximations becomes necessary},\footnote{In fact, for any constant $\varepsilon > 0$ approximating maximum independent set to within a factor of $n^{1-\varepsilon}$ for $n$-node graphs remains NP-hard~\cite{zuckerman07}.
This is one of the reasons why in this paper we focus on planar graphs.} as otherwise we cannot even solve the problem for $k=1$ in polynomial time unless P = NP.

Thus, we consider the setting where we have a quality function $ \sigma:\cF\to \RR_{\geq0}$ assigning a value to each feasible solution, and a niceness factor $c\in [0,1]$.
For maximization problems, we say that a solution is \emph{$c$-optimal} if its objective is at least $c\cdot \max_{S \in \cF} \sigma(S)$. Similarly, for minimization problems, we say that a solution is $c$-optimal if its objective is at most $({1}/{c}) \cdot \min_{S \in \cF} \sigma(S)$, where $1/0$ is treated as $\infty$.
In this work, we assume that an optimization problem is a maximization problem unless otherwise stated.

Now, we give the formal definition for \emph{diverse and nice optimization} (DNO) as follows.
\begin{definition}[Diverse and Nice Optimization]\label{def:dno}
    The input is a five-tuple $(I,\cF,k,\sigma,c)$, where $I$ is a ground set with $n\coloneqq|I|$, $\cF \subseteq 2^I$ is an implicitly prescribed collection of feasible solutions, $k\geq1$ is an integer, $\sigma:2^I\rightarrow\mathbb{R}_{\geq0}$ is a quality function, and $c\in[0,1]$ is a niceness factor.
    Let $ \cF_c =\{ S\subseteq \cF: \text{$S$ is $c$-optimal} \} $.
    The \emph{diverse and nice optimization} problem asks to find a multiset $ \cS=\{S_1,\ldots,S_k\} \in \cF^k_c $ of size $k$ so as to maximize $\sum_{i<j}|S_i\bigtriangleup S_j|$.
\end{definition}

\paragraph{Remark.} 1. There are \textit{two} measures in the problem statement: the \textit{diversity} measure defining the diversity of a collection of $k$ solutions, and a \textit{quality} measure prescribing a required quality for every solution. Another feature of the problem formulation is that it allows one to exploit the quality-diversity tradeoff; note that $c$ is input by the user, unlike in approximation algorithms where one would like $c$ to be as close to $1$ as possible.

2. For a problem $\Pi$ for which there is a $\operatorname{poly}(n,k)$-time \emph{$k$-best enumeration} algorithm, one can decide in $\operatorname{poly}(n,k)$-time whether $\Pi$ has $k$ distinct feasible solutions by simply computing a collection $ \cS = \{S_1, \ldots, S_k\} $ of $k$-best solutions, i.e., $S_1, \ldots, S_k$ such that $ \sigma(S_1) \geq \cdots \geq \sigma(S_k) \geq \sigma(S) $ for all $ S \in \cF-\cS $ (hence the name $k$-best). However, this is not possible anymore in polynomial time if the problem is NP-complete. 
Therefore, in order for the diverse and nice optimization to be well-defined, we allow a multiset of solutions.

\paragraph{Motivation.} A natural motivation is to present the end user with a diverse set of solutions to choose from according to some, perhaps unknown to the algorithm designer, preference. For example, \cite{baste2019fpt} mentions the task of \emph{generating different floor plans for an architect}. With some simplifications, this involves diverse solutions to the \emph{rectangle packing knapsack problem}, where we want to pack diverse sets of high-value items in a large box. To this end, diversification of the simpler, \emph{1D knapsack} problem seems a natural first step. More reasons for diversification include  robustness, reliability, and security~\cite{gao2022obtaining}, algorithmic fairness~\cite{aumuller2020fair}, and portfolio optimization~\cite{drygala2025data}.

Another exciting motivation for diversification is \emph{energy efficiency}: objects such as a dominating set in unit-disk graphs are very useful in monitoring a sensor network. Heuristics for computing a domatic partition (partitioning the vertices of the graph into dominating sets) were presented in~\cite{islam2009maximizing}. The idea is to alternate different dominating sets on a sleep-wake cycle, thereby minimizing the load. This motivates the problem of generating diverse dominating sets in unit disk graphs.\footnote{One may also want to \emph{route} on the set of active sensors, motivating the problem of diverse \emph{connected} dominating sets (CDS). Algorithms to find one CDS were developed in the influential work of ~\cite{demaine2008bidimensionality}.}

However, we observe that a) historically, \emph{planar graphs} have been a natural class to study before unit-disk graphs, and b) the \emph{MWIS} and minimum weight vertex cover (\emph{MWVC}) problems \emph{in planar graphs} have received considerable attention too. In this paper we focus on the diverse variants of the following packing and covering problems. 

\begin{restatable}[\DK]{definition}{ResDiverseKnapsackDefinition}\label{def: diverse-knapsack}
    Let $I=[n]$ be a set of $n$ items, where each $ i \in I$ has integer weight $w_i\geq1$ and profit $p_i\geq1$.
    Let $k\geq1$ be an integer, let $W\geq1$ be the knapsack capacity, and let $c\in[0,1]$ be a niceness factor.
    Let $ \cF = \{ S \subseteq I: w(S) \leq W \} $, where $ w(S) = \sum_{i \in S}w_i $.
    Define the quality $\sigma$ of each feasible solution as the sum of all profits in the solution, i.e., $ \sigma(S) = p(S) = \sum_{i\in S}p_i $.
    The \DK problem asks to find a multiset $\cS=\{S_1,\ldots,S_k\} \in \cF^k_c $ of size $k$ that maximizes $\sum_{i<j}|S_i\bigtriangleup S_j|$.
\end{restatable}

\begin{definition}[\dmwisPG and \dmwvcPG]\label{def: dmwis-pg}
    Let $G=(V,E)$ be a vertex-weighted planar graph with weight function $w\colon V\to\mathbb{R}_{\geq 0}$, and let $k\geq 1$ be an integer, and
let $c\in[0,1]$ be a niceness factor. Let $ \cF $ denote the collection of all independent sets of $G$, and define the quality of $S\in\mathcal{F}$ by $ \sigma(S)=w(S)=\sum_{v\in S}w(v) $.
The \dmwisPG problem asks for a multiset $\mathcal{S}=(S_1,\ldots,S_k)\in\mathcal{F}_c^k $ of $k$ $c$-optimal independent sets that maximizes $ \sum_{i < j} |S_i \triangle S_j|$.
The \dmwvcPG problem is defined analogously.
\end{definition}

\subsection{Related Work and Challenges}
Despite considerable research, there are \emph{no polynomial-time algorithms} with approximation guarantees known for the diverse variants of \emph{any NP-hard optimization problem}.
All existing results either give polynomial (in $n$ and $k$) time approximations for problems solvable in polynomial time, or FPT algorithms (that are exponential in $n$ in the worst case) for NP-hard problems. We highlight some relevant ones below.

For knapsack, the concept of diversity has been explored \emph{within a solution} in~\cite{galvez2022approximation}, where items have colors, and one wants to pack solutions satisfying certain diversity constraints on the color distribution of the items in the solution. However, we have not seen any work on obtaining diverse collections of knapsack solutions, nor for the geometric knapsack versions.

Finding diverse MWISs and MWVCs has, on the other hand, received considerable attention with several FPT results. This line of work focuses on $c=1$, i.e., algorithms that return \emph{optimal} solutions and maximize the diversity \emph{exactly}. While it is clear that such algorithms cannot be polynomial in $n$ and $k$ for NP-complete problems like MWISs in (planar) graphs~\cite{garey1979computers}, even the problem such as finding a diverse pair of maximum matchings is hard~\cite{suomela2010}. Thus the work in this area focuses on fixed-parameter-tractable algorithms that avoid an exponential dependence on the input size $n$; see, e.g.,~\cite{baste2019fpt,baste2022diversity,eiben2024determinantal,fomin2020diverse,fomin2023diverse,kumabe2025max,shida2024finding,misra2025parameterized,austrin2025algorithms,drabik2026finding}.

The works on the \textsc{Diverse Vertex Cover} problem are the most relevant to us. The algorithm in~\cite{baste2019fpt} runs in time $ 2^{k\psi}n^{O(1)} $  where $\psi$ denotes the \emph{size of each solution} (e.g., the size of a minimum vertex cover or a maximum independent set), and the algorithm in~\cite{baste2022diversity} runs in time $ 2^{\omega k} \psi^{O(k)}n^{O(1)} $, where $\omega$ represents the \emph{treewidth} of the input graph.

While the above $ 2^{k\psi}n^{O(1)} $ and $ 2^{\omega k} \psi^{O(k)}n^{O(1)}$ results are impressive and important in the FPT context, in our setting there is a limitation: planar graphs can have large treewidth and large independent sets or vertex covers, i.e., $\psi$ or $\omega$ could be $n^{\Omega(1)}$. Even if $k=2$, this translates to a runtime of $2^{n^{\Omega(1)}}$, which could be prohibitive for many applications.\footnote{Recalling our motivating example of diverse dominating sets in unit-disk graphs, \cite{islam2009maximizing} shows that the size of a minimum dominating set in a sensor network deployed on a $600 \text{m} \times 600 \text{m}$ square goes from 15 to 35 as the number of nodes increases from 100 to 1000 (Figure 5 in the corresponding paper). Assuming a runtime of $2^{k \psi}$ where $\psi$ is the size of a dominating set, the computational task for generating $k=4$ solutions when each dominating set has a size of $15$ will take at least 5 years on a 5GHz computer.} In addition to scalability issues, since these algorithms cannot take $c$ as input, they cannot exploit the quality-diversity tradeoff.

Two related frameworks for diverse solutions to problems solvable in polynomial time were presented in \cite{hanaka2023framework,gao2022obtaining}. 
Hanaka et al.'s framework \cite{hanaka2023framework} only allows to compute diverse solutions satisfying a cardinality constraint (i.e., $c$ cannot be input by the user) but guarantees $\betak$ of optimal diversity and PTASs for some problems.
Gao et al.'s approach allows the user to specify a $c$, but the guaranteed diversity is half of optimal diversity.
Note that not only do we want the best of both worlds ($c$ as input, and $\betak$ as a diversity factor), but, more importantly, we also want our framework to apply to NP-hard problems.

\paragraph{Is diversification possible without exploring the full space of solutions?} There is a good reason for the lack of polynomial time apx. algorithms for diversification of NP-hard problems. Since such an algorithm has to return $k$ many $c$-optimal solutions, an obvious foundation is some existing apx. algorithm for obtaining \textit{one} solution. Traditionally, such algorithms have worked by simplifying the input or ignoring parts of it, while still allowing for good approximation guarantees. Baker's strategy on planar graphs (See \Cref{sec:DIS-PG} for details), or other ``thinning'' or ``sparsifying'' procedures are standard in such algorithms. Therefore, existing algorithms for finding one solution restrict the solutions to have a simpler, or more tractable form.  

However, simplifying or ignoring parts of the input can have a drastic effect on diversity. This means that existing algorithms for finding one approximate solution cannot be used directly, because the simpler, or more tractable subspace of solutions may not have enough diversity compared to the unrestricted optimal set of $k$ solutions.

\vspace{2mm}In this paper, we give the first bi-criteria PTASs (Polynomial Time Approximation Schemes) for several NP-hard problems such as \textsc{Knapsack}, \textsc{MWIS} and \textsc{MWVC} in planar graphs.
In the process of developing these approximation algorithms, we develop a general framework, that we hope can be applied to other problems of interest in the future.

\section{Results}\label{sec:contribution}

In this section, we introduce our results. We begin with defining the notions of bi-criteria approximation.

\begin{definition}[$(\beta,\alpha)$-Approximation]
    Given an instance of \DNO and $ \alpha, \beta \in [0,1] $, an algorithm is called a \emph{$(\beta, \alpha)$-approximation} (or simply $(\beta,\alpha)$-apx.) algorithm if it returns $k$ feasible solutions $S_1, S_2, \ldots, S_k $ such that 1) each $ S_{i} $, $i\in[k]$, is $(\alpha c)$-optimal, and 2) $ \sum_{i<j}|S_i \triangle S_j| \geq \beta \sum_{i<j}|S'_i\triangle S'_j|$ for any collection of $k$ $ c$-optimal solutions $S'_1, \ldots, S'_k$. We refer to $S_1, \ldots, S_k$ as \emph{$(\beta,\alpha)$-diverse solutions}. We refer to $ (\beta,0) $-apx. algorithm simply as $\beta$-apx. algorithm.
\end{definition}

\vspace{2mm} Now, we provide our bi-approximation results.

\begin{restatable}{theorem}{ResDiverseKnapsackThm}{\normalfont[Bi-PTAS for \textsc{Diverse $c$-Knapsack}]}\label{thm: diverse-knapsack}
    For any $ c, \delta, \varepsilon \in (0,1) $, \DK admits a \RunTimeDKLP-time $ (1-\varepsilon, 1-\delta) $-apx. algorithm. In particular, if the output solutions are allowed to have total weight at most $(1+\gamma)W$ for any prescribed $\gamma\in(0,1)$. an $ n^{O(\varepsilon^{-1})}\rho(\delta,\varepsilon,\gamma,c)k^4 $-time bi-PTAS is available for some computable function $\rho$.
\end{restatable}

Note that we provide two results for \DK. The first requires the output solutions to satisfy the weight constraint, whereas the second allows them to violate the constraint by a controlled amount. In exchange, the second algorithm runs slightly faster.

\begin{restatable}{theorem}{RestBiApxISPG}{{\normalfont[Bi-PTASs for \dmwisPG and \dmwvcPG]}}
\label{thm:bi-apx-pg}
    For any $ c, \delta, \varepsilon \in (0,1)$, there exist $2^{O(\delta^{-1}\varepsilon^{-2})}n^{O(\varepsilon^{-1})} k^5 $-time $ (1-\varepsilon, 1-\delta) $-apx. algorithms for \dmwisPG and \dmwvcPG.
\end{restatable}

\paragraph{Remark.} \Cref{thm: diverse-knapsack} and \Cref{thm:bi-apx-pg} are the first examples of bi-criteria PTASs for the diverse variants of \emph{any NP-hard optimization problems}.
This was not possible with existing work even for $k=2$ due to the exponential dependence on other parameters such as the treewidth or the size of an MWIS, as the focus was on exact algorithms (for both diversity and quality). The tradeoff is that we lose the small factors of $ \varepsilon $ in diversity and $ \delta $ in the quality.

\vspace{2mm}\KP~\cite{ibarra1975fast} and MWIS-PG~\cite{baker1994approximation} admit PTASs.
Consequently, for every fixed $c\in(0,1)$, a single $c$-optimal solution can be computed in polynomial time.
This naturally raises the question of whether a diverse collection of $c$-optimal solutions can also be computed in polynomial time.
Surprisingly, assuming $\mathrm{P}\neq\mathrm{NP}$, we show that this is impossible even for $k=2$.
Indeed, we prove that \emph{both} $c$, the solution-quality threshold, and the diversity guarantee must be approximated to achieve a polynomial time algorithm.
Moreover, even under such a simultaneous relaxation, we show that none of the problems admits a bi-criteria FPTAS.

\begin{restatable}{theorem}{Inapprox}{\normalfont{[Inapproximability]}}\label{thm:inapprox}
    Assuming $\PneqNP$, the following hold for \DK and \dmwisPG for any $ c, \delta, \varepsilon \in (0, 1) $ even for $k=2$.
\begin{enumerate}
    \item None of the problems admits either a $ (1-\varepsilon, 1)$- or a $ (1, 1-\delta ) $-apx. algorithm running in polynomial time.
    \item None of the problems admits a $ (1-\varepsilon, 1-\delta)$-apx. algorithm running in time $\operatorname{poly}(n,k, \frac{1}{c},\frac{1}{\delta},  \frac{1}{\varepsilon})$.
\end{enumerate}
\end{restatable}

\paragraph{Other Applications.} 
Our framework extends to other problems, including \textsc{Diverse Enclosing-Polygons} and \textsc{Diverse TSP}, running in time $ O(\delta^{-1} n^5k^3\log k) $ and $ O^*(k^4 2^n) $, respectively.
See \Cref{sec: diverse-enclosing-polygons} and \Cref{sec:diverse-tsp} for details.

\paragraph{Minimum Distance.} The minimum pairwise distance has been a well-studied alternative diversity measure~\cite{eiben2024determinantal,baste2019fpt,fomin2023diverse}. This measure is generally considered more challenging than the sum diversity measure, evidenced by the fact that most of the known results are of the FPT type.
Recently, \cite{chakraborty2026maximizing} presented a randomized polynomial-time approximation algorithm for the diverse median string problem, but the median string problem itself is solvable in polynomial time. Despite considerable research, to the best of our knowledge, no approximation algorithm providing a guarantee on the minimum pairwise distance is known for diverse variants of NP-hard problems.

\section{The Algorithmic Design Framework}\label{sec: framework}

Our algorithm uses a \emph{local-search} algorithm as a subroutine.
To implement the local search algorithm, we define \emph{Budget-Constrained Optimization} ({\BCO}).

\begin{definition}[Budget-Constrained Optimization]\label{def: bco}
    Let $I$, $\cF$, $\sigma$ and $c$ be as in \Cref{def:dno}.
    Let $ \delta \in (0, 1) $, and let $\cF_{(1-\delta)c} = \{ S \in \cF: \text{$S$ is $(1-\delta)c$-optimal w.r.t $\sigma$} \}$.
    Let $r:I\rightarrow\mathbb{Z}$ be an objective function (independent of $\sigma$) such that $ r(S) = \sum_{e\in S}r(e) $.
    The problem ${\BCO}(I,\sigma,r,\cF'_c)$
    takes as input (implicitly) a space of solutions $\cF'_c \subseteq \cF_{(1-\delta)c}$ and 
    returns $ \argmax_{S \in \cF'} r(S)$.
\end{definition}

We are now in a position to formally state our framework.
\begin{restatable}[Framework]{theorem}{ResFramework}\label{thm:framework}
Let $(I,\cF,k,\sigma,c)$ be an input to the {\DNO} problem, let $\cF_c$ be the space of $c$-optimal solutions, and let $\mathrm{OPT}_\mathrm{div}(\cF')$ denote the maximum diversity of $k$ solutions in $\cF'$ for all $ \cF'\subseteq\cF $.
Given $\delta, \varepsilon \in (0,1)$,
suppose that there exist $ \cF'_c, \cF''_c \subseteq \cF $ such that the following hold:
\begin{enumerate}[label=(\arabic*)]
    \item {\bf (Quality Relaxation)} $\cF'_c \subseteq \cF_{(1-\delta)c} $ and $ \cF''_c \subseteq \cF_{(1-\frac{\delta}{2})c} $.\label{framework-1}
    \item {\bf (Diversity Preservation)} There exist $ S'_1,\ldots,S'_k \in \cF'_c$ and $ S''_1,\ldots,S''_k \in \cF''_c$ each collection with diversity at least
    $ ( 1 -\frac{\varepsilon}{2} ) \cdot  \mathrm{OPT}_\mathrm{div}(\cF_c) $.\label{framework-2}
    \item Let $r:I\rightarrow\mathbb{Z}$ be as given in \Cref{def: bco}. Given $R \geq 0$, if $ r(S) \in [-R, R] $ for any $ S \subseteq I$, then there exists an algorithm for $ {\BCO}(I,\sigma,r,\cF'_c)$ running in time $f'(n, \sigma, R, c, \delta, \varepsilon)$.\label{framework-kbest}
    \item There exists a $f''(n,k,\sigma,c, \delta,\varepsilon)$-time algorithm that returns $S_1, \ldots, S_k \in \cF_{(1-\delta)c} $ such that $ \sum_{i<j}|S_i \triangle S_j| \geq (1-\frac{\varepsilon}{2}) \mathrm{OPT}_{\mathrm{div}}(\cF''_c)  $.\label{framework-exact}
\end{enumerate}
Then, there exists an $O( \max\{(f'(n,\sigma,nk,c,\delta,\varepsilon)k^2\log k),\allowbreak 
\allowbreak f''(n,\frac{4}{\varepsilon},\sigma,c,\delta,\varepsilon)\})$-time $(1-\varepsilon, 1-\delta)$-apx. algorithm for \DNO.
\end{restatable}
\noindent\textbf{Remarks.} Conditions \ref{framework-kbest} and \ref{framework-exact} merge the approaches in \cite{hanaka2023framework} and \cite{gao2022obtaining}, whereas Conditions \ref{framework-1} and \ref{framework-2} enable the generalization to NP-hard problems. 
The framework in \cite{hanaka2023framework} is the special case of the above theorem when each element has a unit weight and $\delta=0$, and therefore $\cF'_c = \cF''_c = \cF $.
Conditions \ref{framework-1} and \ref{framework-2} are trivially satisfied in this setting.
It will turn out that for all the problems we explore in this paper, $f'$ and $f''$ will not have any dependence on $c$.

\vspace{2mm}We first provide our bi-PTAS algorithm for \DNO, and then prove \Cref{thm:framework}.

\subsection{Our Algorithm}\label{sec:algorithm}

For brevity, we henceforth write $ A - a + b \coloneqq (A - \{a\}) \cup \{b\} $ for any multiset $A$, where $A-\{a\}$ removes one occurrence of $a$ and $\cup$ denotes the multiset union.

One of the main goals of our algorithm is to efficiently implement the local search algorithm by Cevallos et al.~\cite{cevallos2019improved}.
Given a search space $ \cX $, this local search algorithm starts with an arbitrary subset $ X = \{x_1,\ldots, x_k\} \subseteq \cX $ of $k$ elements, and then finds a pair of elements $x_{\text{in}} \in X $ and $ x_{\text{out}} \in \cX-X $ such that exchanging these two maximizes the diversity in $ X $, i.e., $ (x_{\text{in}}, x_{\text{out}}) = \argmax_{ (x', x'')} \sum_{ x_i, x_j \in (X - x' + x'') }|x_i\bigtriangleup x_j  | $, where $ x' \in X $  and $ x'' \in \cX - X $.
Indeed, they prove that $O(k \log k)$ swaps are enough to obtain $k$ solutions with diversity at least $\betak$-fraction of maximum diversity.

Note, however, that this algorithm might run in time $\exp(n)$ if the search space is implicitly given (e.g., spanning tree). 
Hanaka et al. (2023) showed a strategy to overcome the aforementioned issue. Their approach is the following.
Let $\cF$ be the family of all feasible solutions. Given $ \cS = \{S_1,\ldots,S_k\} \subseteq \cF $, define an objective function $r(\cdot)$ as
\begin{equation}\label{eq: defi-r}
    r_{\cS-S_{\text{in}}}(e) = \sum_{ S \in {\cS-S_{\text{in}}} }(\mathbf{1}(e\not\in S) - \mathbf{1}(e\in S)),
\end{equation}
and $r(S) = \sum_{e\in S} r(e)$.
Intuitively, a feasible subset $S$ with a large value of $r(S)$ tends to contain elements underrepresented in $\cS$.
Then, they argue that finding $ S_{\text{in}} \in \cS $ and $ S_{\text{out}} \in \cF $ that maximize the diversity of $ \cS - S_{\text{in}} + S_{\text{out}} $ is equivalent to finding \emph{farthest} solution from $ \cS - S_{\text{in}} $, i.e.,
\begin{equation}
    S_{\text{out}} = \argmax_{S \in \cF }\sum_{ S' \in \cS -S_{\text{in}}}|S \bigtriangleup S'|.
\end{equation}
In other words, the diverse optimization problem can be approximated via maximization problem. Therefore, with correctly designed farthest insertion algorithm, by making at most $k$ guesses for $S_{\text{in}}$ we may compute the best swapping pair $ (S_{\text{in}}, S_{\text{out}}) $.
Since $O(k\log k)$ swaps suffice to reach the promised diversity, it follows that \emph{by executing farthest insertion $O(k^2\log k)$ times, we may compute a collection of solutions with diversity at least $\betak$ of maximum.}

We note here that the approach of Cevallos et al. (2019) still holds when duplicate solutions are allowed. Their local search algorithm relies on the fact that the collection of all subsets of $\cX$ with size at most $k$ forms a matroid. One can duplicate every element in $\cX$ $k$ times, and the same argument holds for the collection $\cX'$ with the duplicates.

We are now ready to present our algorithm. We provide our algorithm in \Cref{alg:diverse} as a pseudocode.
In \Cref{alg:diverse}, for a multiset $ \mathcal{S} = \{S_1,\ldots, S_k\} $,
$ \mathrm{div}(\cS) $ denotes $ \sum_{i< j}|S_i \bigtriangleup S_j| $.
\vspace{-0.5\baselineskip}

\begingroup
\setlength{\intextsep}{6pt}
\setlength{\textfloatsep}{6pt}
\setlength{\floatsep}{3pt}

\begin{algorithm}[h]
\caption{Diverse and Nice Optimization (\DNO)}
\label{alg:diverse}

\begin{algorithmic}[1]
\REQUIRE $I, \cF, k, \sigma $, $ c$, $\delta$ and $\varepsilon$

\IF{$k \geq 4/\varepsilon$}
    \STATE
    $\mathcal{S} \leftarrow \{S_1,\ldots,S_k\}$,
    where $S_1=\cdots=S_k \in \cF'_c$ can be obtained by BCO by setting $r(e)=0 \quad \forall e\in I$

    \FOR{$h \leftarrow 1$ \TO $\lceil 3k\ln(k+2)\rceil$}
        \FOR{$i\leftarrow1$ \TO $k$}
            \STATE $ S^{\mathrm{in}}_i \leftarrow \text{the $i$-th occurrence of }\cS $
                \FOR{each $e \in I$}
                    \STATE $\displaystyle r(e) \leftarrow \sum_{S \in \mathcal{S} - S^{\mathrm{in}}_i}
                    \left(
                        \mathbf{1}(e \notin S)
                        -
                        \mathbf{1}(e \in S)
                    \right)$
                \ENDFOR

            \STATE
            $S^{\mathrm{out}}_i
            \leftarrow
            \cA_{\BCO}
            (I,\sigma,r,\cF'_c)$
            \COMMENT{Condition \ref{framework-kbest}}
        \ENDFOR

        \STATE
        $\displaystyle
        (S_{\mathrm{in}}^*,S_{\mathrm{out}}^*)
        \leftarrow
        \argmax_{(S^{\mathrm{in}}_i, S^{\mathrm{out}}_i)}
        \mathrm{div}(\mathcal{S}
            - S^{\mathrm{in}}_i
            + S^{\mathrm{out}}_i)$

        \IF{$
            \mathrm{div}(\mathcal{S}
                - S_{\mathrm{in}}^*
                + S_{\mathrm{out}}^*)
            >
            \mathrm{div}(\mathcal{S})
        $}
            \STATE
            $\mathcal{S}
            \leftarrow
            \mathcal{S}
            - S_{\mathrm{in}}^*
            + S_{\mathrm{out}}^*$
        \ELSE
            \RETURN $\mathcal{S}$
        \ENDIF
    \ENDFOR

    \RETURN $\mathcal{S}$

\ELSE
    \STATE
    $\mathcal{S} \leftarrow \{S_1,\ldots,S_k\}$, where each $S_i$ is $ (1-\delta)c $-optimal and $ \sum_{i<j}| S_i \triangle S_j | \geq (1-\frac{\varepsilon}{2})\mathrm{OPT}_{\mathrm{div}}(\cF''_c) $
    \COMMENT{C. \ref{framework-exact}}

    \RETURN $\mathcal{S}$
\ENDIF

\end{algorithmic}
\end{algorithm}

\endgroup

\begin{proof}[Proof of \Cref{thm:framework}]
    Let $(I,\cF,k,\sigma,c)$ be an input to the \DNO, $\delta,\varepsilon \in (0,1)$, and assume that all four conditions are satisfied.
    Assume that $k \geq {4}/{\varepsilon}$.
    By Condition \ref{framework-kbest}, we may suppose that there is an algorithm $\cA_{\BCO}$ that solves the associated $ {\BCO}(I,\sigma,r,\cF'_c)$ running in time $f'(n,\sigma, R, c, \delta, \varepsilon )$ for an integer valued function $r(\cdot)$ bounded by $R$.
    We then implement the local search algorithm we mentioned above according to Hanaka et al.'s approach by using the algorithm $\cA_{\BCO}$.
    To this end, we begin with $k$ copies of a solution $ S^{(0)} \in \cF'_c $.
    Define $ r(\cdot) $ as in \Cref{eq: defi-r}.
    Since $r(e)$ denotes the number of solutions which do not contain $e$ minus the number of solutions that contain $e$ and there are at most $k$ such solutions, $|r(S)|$ for any $S\subseteq I$ is at most $nk$.
    Thus the running time of the $\cA$ algorithm is at most $f'(n,\sigma,nk, c, \delta, \varepsilon)$.
    Therefore, in time $O(f'(n,\sigma,nk,c, \delta,\varepsilon)k^2\log k)$ we can find solutions $S_1,\ldots,S_k$ in $\cF'_c$ such that $\sum_{i<j}|S_i \bigtriangleup S_j| \geq \betak \cdot \mathrm{OPT}_{\mathrm{div}}(\cF'_c)$. Since $ k \geq \frac{4}{\varepsilon} $ and $\cF'_c$ contains solutions $ S'_1,\ldots,S'_k $ such that $ \sum_{i<j}|S'_i\bigtriangleup S'_j| \geq (1-\varepsilon/2)\mathrm{OPT}_{\mathrm{div}}(\cF_c) $, we have
    \begin{equation*}
    \begin{split}
        \textstyle\sum_{i<j}|S_i \bigtriangleup S_j| & \textstyle\geq \left(1-\frac{2}{k+1}\right)\left(1-\frac{\varepsilon}{2}\right)\mathrm{OPT}_{\mathrm{div}}(\cF_c)\\
        &\textstyle\geq (1-\varepsilon)\mathrm{OPT}_{\mathrm{div}}(\cF_c).
    \end{split}
    \end{equation*}

    When $k<\frac{4}{\varepsilon}$, run the $ f''(n,4/\varepsilon,\sigma,c,\delta,\varepsilon) $-time algorithm according to Condition \ref{framework-exact}. 
    Let $S_1,\ldots,S_k $ be the returned solutions by this algorithm.
    Then, each $S_i$ is $ (1-\delta)c $-optimal and $S_1, \ldots, S_k$ have diversity at least $ (1-\frac{\varepsilon}{2})\mathrm{OPT}_{\mathrm{div}}(\cF''_c) $, arguing similarly to the case $ k \geq \frac{4}{\varepsilon} $, we have
    $ \sum_{i<j}|S_i\bigtriangleup S_j| \geq (1-\varepsilon)\mathrm{OPT}_{\mathrm{div}}(\cF_c) $, as we desired.
\end{proof}

If one cannot satisfy Condition \ref{framework-2} for $\cF''_c$ and Condition \ref{framework-exact}, we still have the following weaker result whose proof is an immediate consequence of \Cref{thm:framework}.

\begin{corollary}\label{cor:framework}
    Assume the setting of \Cref{thm:framework}, and that all conditions except Condition~\ref{framework-2} for $\cF''_c$ and Condition~\ref{framework-exact} can be satisfied. 
    Then there exists an $O(f'(n,\sigma, nk, c,\delta,\varepsilon)\cdot k^2\log k)$-time $ (\betak(1-\frac{\varepsilon}{2}), 1-\delta)$-apx. algorithm for \DNO. This is a bi-PTAS when $ k \geq \frac{4}{\varepsilon} $.
\end{corollary}

\begin{center}
    \framebox{\parbox{\linewidth}
    {
    {\begin{center}\textsc{Bi-Criteria PTAS Framework for DNO}\end{center}}
    
    \medskip
    
    Given $(I, \cF, k, \sigma, c)$, and $\delta,\varepsilon > 0$,
      \begin{enumerate}[leftmargin=*]
        \item Identify $\cF'_c, \cF''_c$ for Conditions \ref{framework-1} and \ref{framework-2}.
        \item Develop a $\mathrm{poly}(n,k)$-time {\BCO} algorithm in $\cF_c'$.
        \item Develop a $ n^{g(k)} $-time algorithm that returns $ (1-\delta)c $-optimal solutions with diversity at least $ (1-\frac{\varepsilon}{2})\mathrm{OPT}_{\mathrm{div}}(\cF''_c) $.
        \item Apply \Cref{thm:framework}.
      \end{enumerate}
    }
    }
    \end{center}

\vspace{2mm}\noindent In the remainder of the paper, we prove details for our results.
We prove our hardness results first, then provide details for our bi-approximation algorithms.

\section{Proof of Inapproximability (\Cref{thm:inapprox})}\label{sec:hardness}

In this section we prove our inapproximability result.

\Inapprox*

\medskip We prove \Cref{thm:inapprox} over multiple lemmas and corollaries.
We begin with proving $(1,1-\delta)$-inapproximability of \DK.
\begin{lemma}[$(1,1-\delta)$-Inapproximability of \DK]\label{lem:hard-one-delta-DK}
    For any fixed constants $ c, \delta \in (0, 1) $, \DK does not admit a polynomial time $(1,1-\delta)$-approximation algorithm even when $k=2$ under $\PneqNP$.
\end{lemma}
\begin{proof}
    Fix $c, \delta\in(0,1)$.
    We shall show that given a \textsc{Knapsack} instance it is NP-complete to decide if there are two feasible solutions that partition $ I = [n]$. This is enough since any $c(1-\delta)$-optimal solution is a feasible solution. 
    
    We reduce from \textsc{Partition} problem, which is NP-complete~\cite{Karp72}.
    Given a set of $n$ positive integers $ \{ a_1, a_2, \ldots, a_n \}$ such that $ \sum_{i\in[n]}a_i = 2T $ for some integer $T>0$, \textsc{Partition} asks if there is a partition $ \{J_1, J_2\} $ of $ [n] $ such that $ \sum_{j \in J_1}a_j = \sum_{j \in J_2}a_j = T $.
    Given an instance of \textsc{Partition}, define an instance of \textsc{Knapsack} with $n$ items with profits and weights $ p_i = w_i = a_i $ and weight capacity $ W = T $.

    We claim that \textsc{Partition} is a yes instance \emph{if and only if} the constructed \KP instance has two feasible solutions $S_1$ and $S_2$ such that $ |S_1 \triangle S_2| = n $.

    $(\Longrightarrow)$ Assume that the \textsc{Partition} instance is a yes instance. Then, there is a partition $ \{J_1, J_2\} $ of $[n]$ such that $ \sum_{j \in J_1}a_j = \sum_{j \in J_2}a_j = T $. Since $ w_i = a_i $ for every $i\in[n]$ and $W=T$, $ \sum_{j \in J_1}w_j = \sum_{j \in J_1}p_j = \sum_{j \in J_2}w_j = \sum_{j \in J_2}p_j = W $. Therefore, $J_1$ and $J_2$ are two feasible solutions to \KP with distance exactly $n$.
    
    $(\Longleftarrow)$ Now, suppose that the constructed \KP instance has two feasible solutions that partition the set of items. Let $ S_1 $ and $S_2$ be such feasible solutions. By feasibility, $ w(S_1)\leq T $ and $w(S_2) \leq T$. Also, by construction, $ p(S_1) = w(S_1) $ and $p(S_2) = w(S_2)$. If $p(S_1) < T $ and $p(S_2) < T$, then this is contradiction to the assumption that $\{S_1, S_2\}$ is a partition of $[n]$. Therefore, one of $ p(S_1) $ and $p(S_2)$ must be at least $T$. Then, the profit of the other is also at least $T$ since $\{S_1, S_2\}$ forms a partition.
    Therefore, $p(S_1) \geq T $ and $p(S_2) \geq T$, and thus, $ w(S_1) = p(S_1) = w(S_2) = p(S_2) = T $. This shows that the given \textsc{Partition} instance is a yes instance since $ \sum_{i\in S_1} a_i = \sum_{i\in S_2}a_i = T $.
\end{proof}

Note that bi-inapproximability of \DK in fully polynomial time follows immediately from \Cref{lem:hard-one-delta-DK} by putting $ \varepsilon = \frac{1}{2n} $.
Then, for any yes instance of \textsc{Partition}, the bi-approximation algorithm must return two feasible solutions with distance at least $ (1-\varepsilon)n = (1-\frac{1}{2n})n =  n - \frac12 $. Since the distance is integral, the approximation algorithm must return $n$.
Therefore, we obtain the following corollary.
\begin{corollary}[Bi-Inapproximability of \DK in Fully Polynomial Time]\label{co:hard-fptas-dk}
    Assuming $\PneqNP$, \DK does not admit a $ (1-\varepsilon, 1-\delta)$-approximation algorithm running in time $\operatorname{poly}(n,k, \frac{1}{c},\frac{1}{\delta},  \frac{1}{\varepsilon})$.
\end{corollary}

Now, we show $ (1-\varepsilon, 1) $-inapproximability of \DK.
\begin{lemma}[$(1-\varepsilon,1)$-Inapproximability of \DK]\label{lem:hard-one-eps-DK}
    For any $c, \varepsilon \in (0,1)$, \DK does not admit a polynomial time $(1-\varepsilon, 1)$-approximation algorithm even when $k=2$ under $\PneqNP$.
\end{lemma}
\begin{proof}
    We again reduce from \textsc{Partition}.
    Fix $ c,\varepsilon \in (0,1) $.
    Let $ \{a_1, a_2, \ldots, a_n\} $ with be an instance of \textsc{Partition} such that $ \sum_{i\in[n]} a_i = 2T $.
    Construct a \KP instance as in \Cref{lem:hard-one-delta-DK}, with an additional dummy item $ \mathrm{d} $ with profit $p_d = \lfloor \frac{W}{c} \rfloor $ and weight $ w_d = W $.
    We refer to items $[n]$ as \emph{regular} items. Note that any feasible solution to the constructed \KP instance that does not contain the dummy item has profit at most $W$, otherwise, it violates the weight capacity constraint. We call these solutions \emph{regular} solutions. On the other hand, the singleton $\{\mathrm{d}\}$ is feasible and has profit at least $W$. Therefore, $ \{\mathrm{d}\}$ is an optimal solution. We call this solution as a dummy solution.
    
    We claim that any $c$-optimal regular solution has profit exactly $W$. It suffices to show that any $c$-optimal regular solution has profit at least $W$. Let $S'$ be a $c$-optimal regular solution. Then, $ p(S') \geq c\lfloor \frac{W}{c} \rfloor > c \cdot \left( \frac{W}{c} - 1 \right) = W - c > W - 1 $. Since the profit is integral, $p(S') = W $. Therefore, the claim follows.

    Let $ \diam_c(\mathrm{KNP}) $ denote the maximum symmetric difference between two $c$-optimal solutions to the \KP instance.
    We claim that the given \textsc{Partition} instance is a yes instance \emph{if and only if} $ \diam_c(\mathrm{KNP}) > 0 $.

    ($\Longrightarrow$) Assume that the \textsc{Partition} instance is a yes instance. Then, there exists a partition $ \{J_1, J_2\} $ of $[n]$ such that $ \sum_{j \in J_1}a_j = \sum_{j \in J_2} a_j = W $.
    Therefore, $ \sum_{j \in J_1} p_j = \sum_{j \in J_1} w_j = \sum_{j \in J_2} p_j = \sum_{j \in J_2} w_j $.
    Therefore, by the claim above, $J_1$ and $J_2$ are two $c$-optimal solutions to the \KP instance. Moreover, since $J_1$ and $J_2$ form a partition of $[n]$, $ |J_1 \triangle J_2 | = n $.
    
    ($\Longleftarrow$) We prove the contrapositive statement. 
    Assume for contradiction that the \textsc{Partition} instance is a no instance.
    Then, for any subset $J\subseteq[n]$, $ \sum_{j \in J} p_j \neq W $ and $ \sum_{j \in J}w_i \neq W $. Therefore, the constructed \KP instance has a unique $c$-optimal solution, which is a dummy solution. Thus, $ \diam_c(\mathrm{KNP}) = |\{\mathrm{d}\}\triangle\{\mathrm{d}\}| = 0 $, hence the claim.

    Since one may not distinguish ``$=0$'' and ``$>0$'' by any approximation algorithm, the theorem statement follows.
\end{proof}

\begin{lemma}[$(1,1-\delta)$-Inapproximability of \dmwisPG]
    For any $c,\delta \in (0,1)$, there is no polynomial-time $ (1, 1-\delta) $-approximation algorithm for \dmwisPG even when $k=2$ unless $\PeqNP$.
\end{lemma}
\begin{proof}
    We reduce from \textsc{Maximum Induced Bipartite Subgraph} (MIBS). Given a graph, the goal of MIBS is to find an induced bipartite subgraph with as many vertices of the given graph as possible.
    This problem is NP-complete even when the graph is planar~\cite{lewis1980node}.
    It is easy to check that \emph{MIBS is equivalent to finding two independent sets with maximum symmetric difference of the graph.} Moreover, since the symmetric difference of any two sets is equal to that of their complements.
    
    Fix $ c, \delta \in (0,1) $. 
    Given a simple connected planar graph $ H$ on $n$ vertices, fix an arbitrary vertex $h_0$ of $H$. Introduce two additional vertices $z^*$ and $ u $ in the outer face of $H$ and add edges $(z^*,u)$ and $(u,h_0)$.
    Assign these vertices weights as follows:
    \begin{equation}\label{eq:hard-delta-weights}
        w(v) = 1 \quad \forall v \in V(H), \qquad w(u)=1 \qquad\text{and} \qquad w(z^*) = \max\left\{\left\lceil\frac{c\,n}{1-c}\right\rceil + 1, \left\lceil \frac{n+1}{c(1-\delta)} \right\rceil + 1 \right\}.
    \end{equation}
    We refer to $u$ and $z^*$, respectively, as the \emph{connector} and \emph{anchor} vertices.
    Now, define $G=(V,E)$ as follows: \[ V(G) = V(H)\cup\{u,z^*\}, \qquad \text{and} \qquad E(G) =E(H)\cup \{(h_0, u),(u, z^*)\}. \]
    In other words, $G$ is a planar graph with path of length 2 connected to $H$.
    Let $\alpha(H)$ denote the independence number of $H$, and let $ \opt(G) $ denote the weight of a maximum weight independent set of $G$. Then, it is easy to check that $ \opt(G) \geq w(z^*) + \alpha(H) $.
    
    We claim that an independent set $S$ of $G$ is a $(1-\delta)c$-maximum weight independent set of $G$ \emph{if and only if} $S$ contains $z^*$.
    
    \medskip(\emph{if}) By \Cref{eq:hard-delta-weights},
    \begin{align*}
        w(z^*) &> \frac{cn}{1-c} & \Longleftrightarrow \\
        (1-c)w(z^*) &> cn & \Longleftrightarrow \\
        w(z^*) &> c(w(z^*) + n). &
    \end{align*}
    Since $ w(z^*) + n \geq w(z^*) + \alpha(H) $, it follows that $ w(z^*) \geq  c\,\opt(G) $. Since the anchor vertex alone has weight $c$-optimal, every independent set of $G$ containing $z^*$ is $c$-optimal.

    \medskip(\emph{only if}) Note that any independent set of $G$ not containing $z^*$ can have weight at most $n + 1$.
    Again, by \Cref{eq:hard-delta-weights}, \[ (1-\delta)c\,\opt(G) \geq (1-\delta)c\, w(z^*) > n + 1. \]
    Therefore, every $(1-\delta)c$-maximum weight independent set of $G$ must contain $z^*$, hence our claim follows.

    \medskip By the claim above, it immediately follows that an independent set of $H$ can be reduced to a $c$-maximum weight independent set of $G$ by simply adding the anchor vertex $z^*$ to it, and conversely, a $c$-maximum weight independent set of $G$ can be reduced to an independent set of $H$ by just removing $z^*$ from it.
    Therefore, two $c$-maximum weight independent sets of $G$ with distance at least $D_c(G)$ is equivalent to find two independent sets of $H$ with maximum symmetric difference. Since any $ c$-maximum weight independent set is $c(1-\delta)$-optimal, our theorem follows.
\end{proof}

Akin to the way we obtain \Cref{co:hard-fptas-dk} from \Cref{lem:hard-one-delta-DK}, we may obtain the bi-inapproximability of \dmwisPG, i.e., put $ \varepsilon = \frac{1}{2N} $, where $N$ is the total number of vertices of $G$.
\begin{corollary}[Bi-Inapproximability of \dmwisPG in Fully Polynomial Time]\label{co:hard-fptas-ind-vc}
    Assuming $\PneqNP$, \dmwisPG does not admit a $ (1-\varepsilon, 1-\delta)$-approximation algorithm running in time $\operatorname{poly}(n,k, \frac{1}{c},\frac{1}{\delta},  \frac{1}{\varepsilon})$.
\end{corollary}

Finally, we prove $ (1-\varepsilon, 1) $-inapproximability of \dmwisPG and \dmwvcPG.
\begin{lemma}[$(1-\varepsilon,1)$-Inapproximability of \dmwisPG]\label{lem:hard-eps-ind}
    For any $ c,\varepsilon \in (0,1) $, there is no polynomial-time $(1-\varepsilon)$-approximation algorithm for \dmwisPG, even when $k=2$, unless $\PeqNP$.
\end{lemma}
\begin{proof}
    We reduce from the following NP-complete problem~\cite{garey1979computers}:
    \begin{center}
    \framebox{\parbox{0.9\linewidth}{
    \textsc{Independent Set in Planar Graphs} ($\textsc{ISPG}(H,q)$)
    
    \medskip
    
    \textbf{Input:} A planar graph $H$ and an integer $q\geq1$.
    
    \textbf{Question:} Does $H$ contain an independent set of size at least $q$?
    }}
    \end{center}
    Fix $ c,\varepsilon \in (0,1) $.
    Given a planar graph $H$, we build our planar graph $G$ on $H$.
    Let $L$ be an integer such that
    \begin{equation}\label{eq:bound-L}
        (1-\varepsilon)(L+1) > n.
    \end{equation}
    Create a star graph with $L+2$ vertices $x$, $z$, $ \ell_1, \ell_2, \ldots, \ell_{L} $, where $z$ is the center and the other vertices are all leaves. Add an edge between a fixed vertex $h_0 \in V(H)$ and $x$.
    Define $G =(V,E)$ such that $ V = V(H) \cup \{ \ell_1, \ldots, \ell_L, x, z\} $ and $ E = E(H) \cup \{(z, x), (x, h_0)\} \cup \{(z, \ell_i): i \in [L] \} $.
    Let $ R $ be an integer such that
    \begin{equation}\label{eq:bound-R}
        (1-c) R > 1+c.
    \end{equation}
    Now, assign each vertex of $G$ a weight as follows:
    \begin{equation*}
    \begin{split}
        w(h) &= R \qquad \forall h \in V(H) \\
        w(x) &= w(\ell_i) = 1 \qquad\forall \ell\in[L] \\
        w(z) &= B\coloneqq \left\lfloor \frac{(1-c)Rq + L}{c} \right\rfloor
    \end{split}
    \end{equation*}

    Let $\alpha$ denote the independence number of $H$, and let $OPT$ denote the weight of maximum weight independent set of $G$. We first note that
    \begin{equation}\label{eq:opt-G}
        OPT = B+R\alpha.
    \end{equation}
    That is, the maximum weight independence set of $G$ is obtained as the union of the maximum independent set of H and the center of the star graph. It is because that any independence set of $G$ that does not contain the center of the star graph can have weight at most $R\alpha + L + 1$, and $ (B + R\alpha) - (R\alpha + L + 1) = B - L - 1 \geq 0 $.
    We also note that
    \begin{equation}\label{eq:bound-B}
        (1-c)R(q-1) + L+1 < c\,B \leq (1-c)Rq + L.
    \end{equation}
    The preceding inequality is true since the difference between the upper bound and the lower bound is at least 2 due to \Cref{eq:bound-R}.

    We refer to any independent set of $G$ that does not contain $z$ as a \emph{center-free} independent set of $G$.
    Now, we claim that $G$ has a center-free $c$-maximum weight independent set \emph{if and only if} $ \alpha \geq q $.

    $(\Longleftarrow)$ Assume that $ \alpha \geq q $.
    Let $M$ denote a maximum independent set of $H$. Then, $ M \cup \{ \ell_1, \ldots, \ell_L \} $ is a center-free independent set of $G$, and $w(M \cup \{ \ell_1, \ldots, \ell_L \}) = R \alpha + L $. Then,
    \begin{equation}
    \begin{split}
        (R\alpha + L) - c\,OPT &= (R\alpha + L) - c(B+R\alpha) \\
        &= (1-c)R\alpha + L-c\,B \\
        &\geq (1-c)Rq + L-c\,B \overset{\eqref{eq:bound-B}}{\geq} 0.
    \end{split}
    \end{equation}

    $(\Longrightarrow)$ Now, suppose that $ \alpha \leq q-1 $. Assume for contradiction that $G$ has a $c$-maximum weight independent set, $I$, that is center-free. Then, $w(I) \leq R\alpha + L + 1$. However, since $ w(I) \geq c\,OPT$,
    \begin{align*}
        R\alpha + L+1 \geq c(B + R\alpha), 
    \end{align*}
    and thus
    \[(1-c)R\alpha + L+1 \geq c\,B. \]
    Since $ \alpha \leq q-1 $, the preceding inequality is reduced to
    \[ (1-c)R(q-1) + L+1 \geq c\,B, \]
    which contradicts with the lower bound of $B$ in \eqref{eq:bound-B}, hence the claim.

    Now, let $ D_c(G) \coloneqq \max\{ |I \triangle I'|: \text{$I$ and $I'$ are $c$-maximum weight independent sets of $G$}  \} $, and let $ \mathcal{A}_{c,\varepsilon}(G) $ be an algorithm that returns $c$-maximum weight independent sets $I_1$ and $I_2$ of $G$ such that $ |I_1 \triangle I_2 | \geq (1-\varepsilon)D_c$.
    We claim that $\mathcal{A}_{c,\varepsilon}(G)$ distinguishes the two cases $ \alpha \geq q $ and $\alpha \leq q-1$.
    If the given $\textsc{ISPG}(H,q)$ is a yes instance, then, by our claim above, $G$ has a center-free $c$-maximum weight independent set, say $I'$. Then, for the independent set $I'' = M \cup \{z\}$ of $G$, $I'$ differs from $I''$ in the star graph by at least $L+1$ vertices.
    Thus, $ D_c(G) \geq L+1 $.
    Therefore, \[ |I_1 \triangle I_2 | \geq (1-\varepsilon)D_c(G) \geq (1-\varepsilon)(L+1) \overset{\eqref{eq:bound-L}}{>} n. \]
    On the other hand, if $\textsc{ISPG}(H,q)$ is a no instance, then $G$ has no center-free $c$-maximum weight independent sets. Thus, both $I_1$ and $I_2$ must contain $z$, and therefore  $D_c(G) \leq n $.

    Hence, with $\mathcal{A}_{c,\varepsilon}(G)$ one can decide $ \textsc{ISPG}(H,q) $, which is impossible in polynomial time unless $\PeqNP$, and this completes the proof.
\end{proof}

\section{Algorithms for \DK}

We verify that each condition in \Cref{thm:framework} is satisfied for \Cref{thm: diverse-knapsack}.
Throughout this section, assume that an instance of \DK together with $\delta, \varepsilon, \gamma \in (0,1)$ are given.
Moreover, without loss of generality, assume that $ w_i \leq W $ for all $i\in[n]$.
Let $\widehat{S}$ be an arbitrary $(1-\frac{\delta}{2})$-optimal solution. Let $ \widehat{P} \coloneqq p(\widehat{S})$. $\widehat{S}$ can be computed without increasing the running time stated in \Cref{thm: diverse-knapsack}~\cite{ibarra1975fast}.
We fix $\widehat{S}$ and its profit $ \widehat{P} $ throughout this section.

We first verify the conditions for the version of the theorem that does not allow solutions to violate the weight constraint, then we prove for the version that allows the output solutions to have weight up to $ (1+\gamma)W $.

\paragraph{Conditions \ref{framework-1} and \ref{framework-2}.}

Recall that $ \cF $ denotes the collection of all solutions whose total weight is at most $W$.
Define
\begin{equation}\label{eq:f-double-prime-knapsack}
    \cF''_c \coloneqq \left\{ S \in \cF : p(S) \geq c\widehat{P} \right\} .
\end{equation}
Since $ \cF_c \subseteq \cF''_c \subseteq \cF_{c(1-\frac{\delta}{2})} $, it is obvious that \Cref{eq:f-double-prime-knapsack} satisfies the conditions~\ref{framework-1} and ~\ref{framework-2} in \Cref{thm:framework}.

To verify the conditions on $\cF'_c$, we need a lemma.

\begin{restatable}{lemma}{ResKnapsackRounding}\label{lem: knapsack-scaling}
    Consider a \KP instance with weights
    $\{w_i\}_{i=1}^n$, profits $\{p_i\}_{i=1}^n$, and capacity $W$.
    Let $ c, \delta,\gamma\in(0,1) $, and let $S$ be any feasible solution to the instance.
    Define
    \[
        \widetilde{P} = \left\lceil \frac{1-\delta}{\delta} \cdot n \right\rceil,\quad
        \widetilde{p}_i = \left\lfloor \frac{\widetilde{P}+n}{ c\,p(S)}\cdot u_i \right\rfloor,\quad
        \widetilde{W} = \left\lceil \frac{1+\gamma}{\gamma} \cdot n \right\rceil,\quad\text{and}\quad
        \widetilde{w}_i = \left\lceil \frac{\widetilde{W}-n}{W}\cdot w_i \right\rceil,
    \]
    where $ \widetilde{w}(S) = \sum_{i\in S}\widetilde{w}_i $ and $ \widetilde{p}(S) = \sum_{i\in S}\widetilde{p}_i $.
    Then, the following hold:
    \begin{enumerate}[leftmargin=*,label=\normalfont(\roman*)]
        \item If $X$ is any $c$-optimal solution, then $ \widetilde{p}(X) \geq \widetilde{P} $ and $ \widetilde{w}(X) \leq \widetilde{W} $.
        \item If $Y$ is any subset of items such that $ \widetilde{p}(Y) \geq \widetilde{P} $ and $\widetilde{w}(Y) \leq \widetilde{W}$, then $p(Y) \geq c(1-\delta)p(S) $ and $w(Y) \leq (1+\gamma)W$.
    \end{enumerate}
\end{restatable}
\begin{proof} We give the proof for profits; the argument for weights is analogous.
Let $X$ be any $c$-optimal solution. Then, $p(X) \geq c\cdot p(S)$. Therefore,
        \[
        \begin{split}
            \widetilde{p}(X)
                &= \sum_{i\in X} \widetilde{p}_i
                = \sum_{i\in X} \left\lfloor \frac{\widetilde{P}+n}{c\cdot p(S)} \cdot p_i \right\rfloor \\
                &\geq \sum_{i\in X} \left(\frac{\widetilde{P}+n}{c\cdot p(S)} \cdot p_i -1\right)
                = \frac{\widetilde{P}+n}{c\cdot p(S)} \left(\sum_{i\in X} p_i \right) - |X| \\
                &\geq \frac{\widetilde{P}+n}{c\cdot p(S)} \cdot p(X) - n \geq \frac{\widetilde{P}+n}{c\cdot p(S)} \cdot c\cdot p(S) - n = \widetilde{P}.
        \end{split}
        \]
        \item Let $Y$ be any subset of items such that $ \widetilde{p}(Y) \geq \widetilde{P} $. Then,
        \[
            \frac{\widetilde{P}+n}{c\cdot p(S)} \cdot p(Y) = \sum_{i\in Y} \frac{\widetilde{P}+n}{c\cdot p(S)}\cdot p_i \geq \sum_{h\in Y} \left\lfloor\frac{\widetilde{P}+n}{c\cdot p(S)}\cdot p_i \right\rfloor = \widetilde{p}(Y) \geq \widetilde{P}.
        \]
        Therefore, it follows that
        \[
        \begin{split}
            p(Y)
                \geq \frac{c \cdot p(S)}{\widetilde{P}+n} \cdot \widetilde{P}
                = \frac{\widetilde{P}}{\widetilde{P}+n}\cdot c\cdot p(S)
                &= \frac{ \left\lceil \frac{(1-\delta)n}{\delta} \right\rceil }{ \left\lceil\frac{(1-\delta)n}{\delta}\right\rceil + n }\cdot c\cdot p(S) \\
                &\geq \frac{ \frac{(1-\delta)n}{\delta} }{ \frac{(1-\delta)n}{\delta} + n } \cdot c \cdot p(S)
                =c(1-\delta)p(S).\qedhere
        \end{split}
        \]
\end{proof}

The following corollary can be proved easily from the previous lemma, and this verifies the conditions on $\cF'_c$.
\begin{restatable}{corollary}{FPrimeCKnapsack}\label{lem:f-prime-c-knapsack}
Given an instance of \KP and $ \delta\in(0,1) $, let $ \widetilde{P} = \left\lceil \frac{(2-\delta)n}{\delta} \right\rceil $ and $ \widetilde{p}_i = \left\lfloor \frac{(\widetilde{P}+n)p_i}{c\cdot \widehat{P}} \right\rfloor $ for each  $i\in[n]$. Define
\begin{equation}\label{eq:f-prime-c-knapsack}
	\cF'_c = \{ S \in \cF: \widetilde{p}(S) \geq \widetilde{P} \}.
\end{equation}
Then, \Cref{eq:f-prime-c-knapsack} satisfies conditions~\ref{framework-1} and ~\ref{framework-2} in \Cref{thm:framework}.
\end{restatable}
\begin{proof}
    Recall that Condition~\ref{framework-1} on $\cF'_c$ says $ \cF'_c \subseteq \cF_{c(1-\delta)} $ and Condition~\ref{framework-2} on $\cF'_c$ says there exist $ S'_1, \ldots S'_k \in \cF'_c $ such that $ \sum_{i<j}|S'_i \triangle S'_j| \geq (1-\frac{\varepsilon}{2})\cdot\mathrm{OPT}_{\mathrm{div}}(\cF_c) $.
    In \Cref{lem: knapsack-scaling}, replace $ \delta $ with $ \frac{\delta}{4} $ and $S$ with $\widehat{S}$.

    Suppose that $ Y \in \cF'_c $.
    Let $S^*$ denote the total profit of an optimal solution to the given \KP instance.
    Then, by \Cref{lem: knapsack-scaling}(ii), it follows that
    \[ p(Y) \geq c(1-\frac{\delta}{2})\widehat{P} \geq c(1-\delta)p(S^*). \]
    where the last inequality is obtained by applying Bernoulli's inequality.
    Therefore, $ \cF'_c \subseteq \cF_{c(1-\delta)} $, and this verifies Condition~\ref{framework-1} in \Cref{thm:framework}.
    To verify Condition~\ref{framework-2}, note that due to \Cref{lem: knapsack-scaling}(i), every $c$-optimal solution satisfies \Cref{eq:f-prime-c-knapsack}. That is, $ \cF_c \subseteq \cF'_c $, and this completes the proof.
\end{proof}

\paragraph{Condition~\ref{framework-kbest}.}
To verify Condition~\ref{framework-kbest} in \Cref{thm:framework}, we give a dynamic programming (DP) algorithm running in $\cF'_c$.
First, consider the following lemma.

\begin{restatable}{lemma}{ResDiverseKnapsackBetak}\label{lem: diverse-knapsack-beta-k}
    Given an instance of \DK and some integers $P,R\geq0$, let $ r: I \rightarrow \mathbb{Z} $ be an objective function such that $r(S) = \sum_{i \in S}r(i)$ ranging in $[-R, R]$.
    Then, there exists an $ O(n(P+1)(R+1)) $-time algorithm that returns $ \argmax\{ r(S): S\in \cF \text{ and }  p(S) \geq P \} $.
\end{restatable}
\begin{proof}
    Define $\text{DP}[i,P',R']$ as the smallest possible total weight of a subset of items in $\{1,\ldots,i\} \subseteq [n]$ whose total profit is at least $P'$ and whose total objective value is exactly $R'$, where $ P' \leq P $.
    For simplicity, denote $r(i)$ by $r_i$ for every $i\in I$.
    Then, $\text{DP}[i,P',R']$ can be defined as the following recurrence relation:
    \begin{equation*}\label{eq: diverse_knapsack_beta_k_rec_rel}
        \text{DP}[i,P',R'] = \min
        \begin{cases}
            \text{DP}[i-1,P',R'] \\
            \text{DP}[i-1,M(i,P'),N(i,R')] + w_i
        \end{cases}
    \end{equation*}
    where $ M(i, P') = \max\{0, P'-p_i\} $ and $ N(i, R') = R'-r_i $ if $ R'-r_i \in [-R,R] $, and otherwise, $\infty$.
    Note that the initial conditions are $ \text{DP}[0,P',R'] = 0 $ if $ (P',R')=(0,0) $, and otherwise, $\infty$.
    Since $ P' \leq P $ and $R'\in[-R,R]$, the entire dynamic programming table can be computed in time $O(n(P+1)\cdot2(R+1)) = O(n(P+1)(R+1))$.
    We then perform a simple linear scan to find the largest $R'$ such that $ \text{DP}[n,P,R'] \leq W $.
    To obtain an actual solution, standard bookkeeping can be used; we omit the details.
\end{proof}
To verify Condition~\ref{framework-kbest}, use the adjusted profit $ \{\widetilde{p}_i\}_{i=1}^n $ and the adjusted profit requirement $\widetilde{P}$ given in \Cref{lem:f-prime-c-knapsack}. 
Note that the running time of this associated BCO algorithm is $ O(n^3k\delta^{-1}) $.

\paragraph{Condition \ref{framework-exact}.}

To verify Condition~\ref{framework-exact}, we develop a \RunTimeDKSmallKLP-time algorithm. Recall that, according to our framework, this subroutine is executed only when $k < \frac{4}{\varepsilon}$; hence, the overall running time remains polynomial.

\begin{lemma}\label{lem: knapsack-small-k-LP}
    \DK admits a \RunTimeDKSmallKLP-time algorithm that returns $k$ feasible solutions, each having profit at least $ (1-\frac{\delta}{2})c\widehat{P} $, whose diversity is at least $ (1 -\frac{\varepsilon}{2})\mathrm{OPT}_{\mathrm{div}}(\cF''_c) $.
\end{lemma}

\noindent Note that \Cref{lem: knapsack-small-k-LP} verifies Condition~\ref{framework-exact} in \Cref{thm:framework} since $\widehat{P}$ is $(1-\frac{\delta}{2})$-optimal and $ (1-\frac{\delta}{2})^2 \geq (1-\delta) $.

To prove \Cref{lem: knapsack-small-k-LP}, we reduce \DK to a variant of
\textsc{Multiple-choice Multidimensional Knapsack}(MMK)~\cite{patt2012vector}.
In MMK, each item has several incarnations, each represented by a multidimensional weight vector, and at most one incarnation may be selected per item.
The goal is to maximize the total profit subject to a capacity constraint in each dimension.
For a fixed number of dimensions, MMK admits a PTAS.
Roughly speaking, the PTAS enumerates all possible choices of high-profit incarnations and, for each choice, formulates the residual problem as an integer program (IP), solves its linear programming (LP) relaxation, and rounds all fractional variables down.
Since MMK contains only packing constraints, this rounding preserves feasibility; the algorithm then returns the most profitable rounded solution over all guesses.

Our setting is not exactly MMK, since \DK has both packing and covering constraints.
In particular, rounding fractional variables down may violate the profit-covering constraints.
We address this issue by guessing a bounded number of high-profit incarnations and capping the profits of all remaining items.
For a fixed solution $S_t$, if an optimal solution contains at least
$\lceil 4k/\delta\rceil$ incarnations with profit at least $\frac{c\delta\widehat{P}}{4k}$, then the guessed incarnations alone satisfy the profit requirement.
Otherwise, all such incarnations can be guessed, and capping does not change the profit contribution of the residual optimal solution.

\begin{proof}[Proof of \Cref{lem: knapsack-small-k-LP}]
For each item $i$ and each nonzero bit pattern $ b\in B\coloneqq \{0,1\}^k\setminus\{0^k\} $, create an incarnation $(i,b)$.
For each bit pattern $b=(b_1,b_2,\ldots,b_k)$ of item $i$ and every $t\in[k]$,
$b_t=1$ means that $i\in S_t$, and $b_t=0$ means that $i\notin S_t$.
Define $ \mathrm{Div}(b) \coloneqq \left|\left\{(t,t') :1\leq t<t'\leq k, b_t\neq b_{t'} \right\} \right| $ for each $ b \in B $. Thus, $\mathrm{Div}(b)$ is the contribution of the incarnation $(i,b)$ to the diversity.
Let $x_{i,b}=1$ indicate that incarnation $(i,b)$ is selected.
Consider the following IP:
\begin{align}
\max\quad
& \sum_{i\in I}\sum_{b\in B}\mathrm{Div}(b)x_{i,b}
\label{eq:dk-ip-obj}\\
\text{s.t.}\quad
& \sum_{i\in I}\sum_{b\in B}w_i b_t x_{i,b}\le W
&& \forall t\in[k]
\label{eq:dk-ip-weight}\\
& \sum_{i\in I}\sum_{b\in B}p_i b_t x_{i,b}\ge c\widehat{P}
&& \forall t\in[k]
\label{eq:dk-ip-profit}\\
& \sum_{b\in B}x_{i,b}\le 1
&& \forall i\in I
\label{eq:dk-ip-choice}\\
& x_{i,b}\in\{0,1\}
&& \forall i\in I,\ \forall b\in B.
\label{eq:dk-ip-integral}
\end{align}
Here, \Cref{eq:dk-ip-weight} and \Cref{eq:dk-ip-profit} are, respectively,
the weight-capacity and profit-requirement constraints for each knapsack, and
\Cref{eq:dk-ip-choice} ensures that at most one incarnation is selected for each
item.
If no incarnation of item $i$ is selected, then none of the solutions contains
item $i$.
For a set $A$ of incarnations, write $ \mathrm{Div}(A)\coloneqq  \sum_{(i,b)\in A}\mathrm{Div}(b) $ and use the analogous notation for a fractional vector.
We give a $(1-\frac{\varepsilon}{2})$-approximation to the diversity objective
while guaranteeing that each resulting solution has profit at least
$c(1-\frac{\delta}{2})\widehat{P}$.

\paragraph{Algorithm.}
Our algorithm is guess-and-LP with rounding.
Let $ x^{\mathrm{Int}} = \{x^{\mathrm{Int}}_{i,b}\}_{i\in I,b\in B} $ be an optimal integral solution to the IP.
Let $ \mathrm{OPT_{div}(IP)}
    \coloneqq  \sum_{i\in I}\sum_{b\in B} \mathrm{Div}(b)x^{\mathrm{Int}}_{i,b} $.
Note that $\mathrm{OPT_{div}(IP)} = \mathrm{OPT_{div}}(\cF''_c)$ by \Cref{eq:f-double-prime-knapsack}.
Now, we present our algorithm.
\begin{enumerate}[label=\arabic*.]
    \item Enumerate every tuple $ (H_D,H_1,\ldots,H_k) $, where $H_D$ is a guess for the highest-diversity subset of $ \mathrm{supp}(x^{\mathrm{Int}}) $ and is a set of at most $\left\lceil\frac{4k}{\varepsilon}\right\rceil$ incarnations, $H_t$ is a set of at most $\left\lceil\frac{4k}{\delta}\right\rceil$ incarnations for every $t\in[k]$, and $ b_t=1 $ and $ p_i\geq\frac{c\delta\widehat{P}}{4k} $ for  every $(i,b)\in H_t$.
    Let $ H\coloneqq H_D\cup H_1\cup\cdots\cup H_k $.
    Discard a guess if $H$ contains two different incarnations of the same item.

    \item For each remaining guess $H$, we set up an integer program on the residual incarnations as follows. Define
    \begin{align*}
        I'&\coloneqq  I\setminus\{i:(i,b)\in H\text{ for some }b\in B\}, & \\
        p'_i &\coloneqq \textstyle\min\left\{p_i,\frac{c\delta\widehat{P}}{4k}\right\} &\forall i \in I' \\
        \widehat{P}'_t &\coloneqq 
        {\textstyle\max\left\{0,c\widehat{P}-\sum_{(i,b)\in H}p_i b_t \right\}}
        &\forall t\in[k] \\
        W'_t &\coloneqq W-{\textstyle\sum_{(i,b)\in H}w_i b_t} &\forall t\in[k]
    \end{align*}
    If $W'_t<0$ for some $t$, discard the guess.
    Define
    \[\textstyle \tau_D \coloneqq \min_{(i,b)\in H_D}\mathrm{Div}(b) \]
    if $ |H_D| = \left\lceil \frac{4k}{\varepsilon} \right\rceil $, and $0$, otherwise.
    Note that we set $\tau_D = 0 $ when $ |H_D| < \left\lceil\frac{4k}{\varepsilon}\right\rceil \ $ since this means that the algorithm is guessing the case $ \mathrm{supp}(x^{\mathrm{Int}}) < \left\lceil\frac{4k}{\varepsilon}\right\rceil $ thus additional diversity is not needed from the residual incarnations.
    Let $ B' \coloneqq \{b\in B:\mathrm{Div}(b)\le\tau_D\} $.
    Finally, define the residual IP as the original IP with $ I$, $p_i$, $c\widehat{P}$, $W$ and $B$, respectively, replaced with $ I'$, $p'_i$, $\widehat{P}'_t$, $W'_t$ and $B'$.

    \item Relax $ x_{i,b}\in\{0,1\} $ to $ x_{i,b} \geq 0 $, and solve the program. Let $x^*$ be an optimal basic feasible solution.

    \item
    Call an item $i\in I'$ \emph{fractional} if its block
    $\{x^*_{i,b}\}_{b\in B'}$ is neither identically zero nor an integral block
    having exactly one variable equal to one.
    Let $F$ be the set of fractional items, and define $ x^R_{i,b} \coloneqq 0  $ if $ i \in F$, and otherwise, $ x^*_{i,b} $, i.e., round down to 0 fractional incarnations.
    Let $\overline{x}$ be the integral solution consisting of the guessed incarnations in $H$ together with the incarnations selected by $x^R$.

    \item Among all non-discarded guesses, return the corresponding solution
    $\overline{x}$ with the greatest diversity.
\end{enumerate}

\paragraph{Running time.} The bottleneck of the algorithm is the enumeration of $H$.
Since each item has at most $2^k$ incarnations, the total running time is \RunTimeDKSmallKLP.

\paragraph{Correctness.}
Since the residual LP has only $2k$ global constraints, an optimal basic feasible solution has at most $2k$ fractional item blocks. Therefore, $|F| \leq 2k$.

Now, we establish the required guarantees. For every non-discarded guess, weight feasibility is satisfied since rounding factional item blocks down to zero does not violate the packing constraint.
Thus, it suffices to show both profit guarantee and diversity guarantee are satisfied.
\begin{enumerate}
    \item \textbf{(Profit guarantee.)} Consider any non-discarded guess and fix $t\in[k]$.
    Since $x^*$ has at most $2k$ fractional item blocks and $p'_i\leq \frac{c\delta\widehat{P}}{4k}$, rounding them down loses at most $\frac{c\delta\widehat{P}}{2}$ capped profit.
Therefore,
\begin{align*}
    \textstyle\sum_{(i,b)\in H} &\textstyle p_i b_t
    + \sum_{i\in I'}\sum_{b\in B'} p_i b_t x^R_{i,b}\\
    &\textstyle\geq
    \sum_{(i,b)\in H} p_i b_t
    + \widehat{P}'_t
    - \frac{c\delta\widehat{P}}{2}
    \geq
    c\left(1-\frac{\delta}{2}\right)\widehat{P},
\end{align*}
where the last inequality holds by the definition of
$\widehat{P}'_t$.

\item \textbf{(Diversity guarantee.)} Consider the guess $ H = H_D \cup H_1 \cup \cdots \cup H_k $ in which $H_D$ contains the $\left\lceil\frac{4k}{\varepsilon}\right\rceil$ highest-diversity incarnations of $x^{\mathrm{Int}}$, or all its incarnations if there are fewer, and each $H_t$ contains all profit-high incarnations assigned to solution $t$, or any $\left\lceil\frac{4k}{\delta}\right\rceil$ of them if there are more, i.e., $H$ is the correct guess.

We claim that the residual part of $x^{\mathrm{Int}}$ is feasible for the corresponding LP.
Indeed, if $H_t$ contains $\left\lceil\frac{4k}{\delta}\right\rceil$ incarnations, then their total profit is at least $c\widehat{P}$ and hence $\widehat{P}'_t=0$.
Otherwise, every residual incarnation assigned to solution $t$ has profit less than $\frac{c\delta\widehat{P}}{4k}$, so capping does not change its profit.

Now, we prove that $\overline{x} $ has the guaranteed diversity. Since every incarnation in $H_D$ has diversity at least $ \tau_D $, $ \left\lceil \frac{4k}{\varepsilon} \right\rceil\tau_D \leq \mathrm{Div}(H_D)  $. Since $ \mathrm{Div}(H_D) \leq \mathrm{OPT}(\mathrm{IP}) $, it follows that
$ \tau_D \leq \frac{\varepsilon\, \mathrm{OPT_{div}(IP)}}{4k}$.
Moreover, every incarnation has diversity at most $ \tau_D $ in the residual LP, thus the algorithm loses at most $ 2k \tau_D $ in diversity at the rounding step.
Therefore,
\begin{align*}
    \textstyle\mathrm{Div}(\overline{x})
    &\geq
    \textstyle\mathrm{Div}(H)+\mathrm{Div}(x^*)-2k\tau_D \\
    &\geq
    \textstyle\mathrm{OPT_{div}(IP)}-\frac{\varepsilon\mathrm{OPT_{div}(IP)}}{2}\\
    &=
    \textstyle\left(1-\frac{\varepsilon}{2}\right)\mathrm{OPT_{div}(IP)}.
\end{align*}
\end{enumerate}
This completes the proof.
\end{proof}

\paragraph{Putting everything together.} 
With all conditions verified, an application of \Cref{thm:framework} proves \Cref{thm: diverse-knapsack}. 
The runtime is $ \textstyle\max\Big\{ (2^{\frac{4}{\varepsilon}}\cdot n)^{O(\frac{4/\varepsilon}{\varepsilon} + \frac{16/\varepsilon^2}{\delta})}, O(n^3k\delta^{-1})\cdot k^2\log k \Big\}$, which simplifies to \RunTimeDKLP.

Now, we prove that a faster bi-PTAS is available if the output solutions are allowed to have total weight at most $(1+\gamma)W$, for any prescribed $\gamma\in(0,1)$.

\begin{lemma}\label{thm:knapsack-weight-violation}
    For any $c,\delta,\varepsilon,\gamma \in (0,1)$, \DK admits a \RunTimeDKDP-time algorithm that outputs $S_1, \ldots, S_k$ such that, for every $i\in[k]$, $ w(S_i) \leq (1+\gamma)W$, $p(S_i) \geq (1-\delta)c\mathrm{OPT}_W$, and $ \sum_{i<j}|S_i \triangle S_j| \geq (1-\varepsilon)\mathrm{OPT}_{\mathrm{div}}(\cF_c) $, where $\mathrm{OPT}_W$ denotes the maximum profit of a solution with weight at most $W$.
    In other words, there is a faster bi-PTAS when we allow a small capacity violation.
\end{lemma}
To prove \Cref{thm:knapsack-weight-violation}, it suffices to provide an algorithm for the case $ k < \frac{\varepsilon}{4} $.
First, we prove a lemma.
\begin{restatable}{lemma}{ResDiverseKnapsackExact}\label{lem: diverse-knapsack-exact}
    Let $\{w_i\}_{i=1}^n$ and $\{p_i\}_{i=1}^n$ be positive weights and profits of $n$ items in a classical knapsack problem with capacity $W>0$.
    Given integers $k\geq1$, $ d_{min} \geq 0$ and $ P > 0 $, there is an algorithm that runs in time $ O\Bigl((n+1) (d_{min}+1)^{\binom{k}{2}} (W+1)^k(P+1)^k2^k k^2 \Bigr) $ and returns $k$ feasible solutions $S_1,\ldots,S_k$ that maximize $\sum_{i \neq j} |S_i \Delta S_j| $ subject to $ |S_i \Delta S_j | \geq d_{min} $ for all $ 1 \leq i < j \leq k $ and $ p(S_i) \geq P $ for every $i\in[k]$, if they exist. Otherwise, the algorithm reports an error without affecting the overall running time.
\end{restatable}
\begin{proof}
Our algorithm uses a dynamic programming approach.
At each step (for a given item), it enumerates all possible $2^k$ assignments of including or excluding that item in each of the $k$ solutions.
At each step, it keeps track of (i) how many items we have considered so far, (ii) how much \emph{distance} for each pair of solutions is still required, (iii) how much \emph{capacity} remains for each of the $k$ knapsacks, and (iv) how much \emph{profit} is still required for each of the $k$ solutions.
We then show how to compute the dynamic programming table (DP) within the desired time bound and reconstruct the actual $k$ solutions.
Define
\begin{equation}\label{eq:knapsack-def}
	\text{DP} \Bigl[ h, \; d_{1,2},\ldots, d_{i,j},\ldots,d_{k-1,k}, \; W'_1,\ldots, W'_m, \ldots, W'_k, \; P'_1,\ldots,P'_m, \ldots,P'_k \Bigr]
\end{equation}
to be the maximum sum of pairwise distances of any $k$ partial solutions $ S'_1,\dots,S'_m, \ldots,S'_k $ such that for every $m\in[k]$ and for all $ 1 \leq i < j \leq k $:
\begin{equation}
    S'_m \subseteq \{1,2,\dots,h\},
    \quad
    \lvert S'_i \,\Delta\, S'_j\rvert \,\ge\, d'_{i,j},
    \quad
    w(S'_m) \,\le\, W'_m,
    \quad\text{and}\quad
    p(S'_m) \,\ge\, P'_m.
\end{equation}
For brevity, let us use $ \{d'_{i<j}\}_{i,j\in[k]} $ for $ d'_{1,2}, \ldots, d_{i,j},\ldots, d'_{k-1,k} $, and similarly for $ \{W'_m\}_{m=1}^k $ and $ \{P'_m\}_{m=1}^k $.
Then, among all possible $ 2^k $ possibilities of assigning the current item to none, some, or all of the $ k $ partial solutions, the maximum total of pairwise distances of the $k$ partial solutions can be obtained as follows:
\begin{equation*}\label{eq: knapsack-recursion}
\begin{split}
    \text{DP} \Bigl[ h, \; &\{d'_{i<j}\}_{i,j\in[k]}, \; \{W'_m\}_{m=1}^k, \; \{P'_m\}_{m=1}^k \Bigr] \\
    &= \max_{\substack{x \in \{0,1\}^k \\ W'_m \geq w_h\cdot x_h \; \forall m\in[k] }} \text{DP} \Bigl[ h-1, \; \{d''_{i<j}\}_{i,j\in[k]}, \; \{W''_m\}_{m=1}^k, \; \{P''_m\}_{m=1}^k \Bigr] + \sum_{i=1}^k \mathbf{1}(x_i\neq x_j),
\end{split}
\end{equation*}
where $x_i$ represents the $i$-th bit of $x$, $ d''_{i,j} = \max\{0, d'_{i,j} - \mathbf{1}(x_i\neq x_j)\} $, $  W''_m = W'_m - (w_h \cdot x_h) $ and $  P''_m = \max\{0, P'_m - (p_h \cdot x_h) \}$.
Here, the distance and the profit are ``clamped'' to zero once the needed distance or the profit has been achieved, meaning that from the next item onward, the program no longer tracks whether the distance or the profit has gone beyond the requirement—it is simply recorded as fully satisfied.
The constraint $W'_m \geq w_h \cdot x_m$ for all $m \in [k]$ ensures that we do not consider assignments that would exceed the weight capacity in any of the $k$ solutions. Such a case would result in a weight-infeasible assignment, where including the $h$-th item in the $m$-th solution causes its total weight to exceed the remaining capacity $W'_m$.

Also, \Cref{eq:knapsack-def} requires some error and base cases to be specified.
When $h = 0$, if $P'm > 0$ for some $m$, or if $d'_{i,j} > 0$ for some $i < j$, we define the DP value in \Cref{eq:knapsack-def} to be $-\infty$, to indicate a profit-infeasible or diversity-infeasible assignment,
and we let the DP value to be zero when $h = 0$ and $P'm = d'{i,j} = 0$ for all $m$ and $i < j$.

Note that the desired total of pairwise distances is then stored properly in the cell
\begin{equation}\label{eq: diverse-knapsack-final-cell}
  \text{DP}[ n, \{d_{min}\}_{1\leq i< j \leq k}, \{W\}_{m=1}^k, \{P\}_{m=1}^k ]. 
\end{equation}
Since each required pairwise distance $d'_{i,j}$ can range from $0$ to $d_{min}$, each remaining capacity $W'_m$ from $0$ to $W$ and similarly for $P'_m$, there are at most $$ (d_{min}+1)^{\binom{k}{2}}(W+1)^{k}(P+1)^{k} $$ possible states for each $h \in \{0,\dots,n\}.$
Since at each state the algorithm considers $2^k$ ways of assigning the current item to the $k$ solutions and each such assignment requires $O(k^2)$ time to update the current state, computing \Cref{eq: diverse-knapsack-final-cell} takes
\[
O\Bigl((n+1) (d_{min}+1)^{\binom{k}{2}} (W+1)^k(P+1)^k2^k k^2 \Bigr).
\]

Finally, we mention that constructing the actual solutions can be done by bookkeeping. In each DP cell, store the chosen bit-vector $x^*$ that yielded the maximum diversity. When we reach the final cell, we follow its stored pointer back to the cell for $h-1$, etc. Each time we see $x^*_m=1$, it means item $h$ was included in solution $S_m$. Tracing back from $h=n$ down to $h=1$ recovers all choices, giving the final subsets $S_1,\dots,S_k$. This completes the construction of $k$ solutions with the desired constraints.
\end{proof}

Now, we are ready to prove \Cref{thm:knapsack-weight-violation}.
Recall that $ \widehat{S} $ is a fixed $ (1-\frac{\delta}{2}) $-optimal solution and $\widehat{P} = p(\widehat{S})$.
\begin{proof}[Proof of \Cref{thm:knapsack-weight-violation}]
    Given an instance to \DK,
    define $ \cF''_c $ as 
    \[ \cF''_c = \{ S \subseteq I : \widetilde{w}(S) \leq \widetilde{W}\quad \text{ and }\quad \widetilde{p}(S) \geq \widetilde{P}  \}. \]

    Let
    \begin{equation}\label{eq:f-double-parameter}
        \widetilde{P} = \left\lceil \frac{2-\delta}{\delta} \cdot n \right\rceil, \quad \widetilde{p}_i = \left\lfloor \frac{\widetilde{P}+n}{c\cdot \widehat{P}}\cdot p_i \right\rfloor, \quad \widetilde{W} = \left\lceil \frac{1+\gamma}{\gamma} \cdot n \right\rceil, \quad \text{ and }\quad \widetilde{w}_i = \left\lceil \frac{\widetilde{W}-n}{W}\cdot w_i \right\rceil
    \end{equation}
    Then, as in \Cref{lem:f-prime-c-knapsack}, one can easily show that $\cF''_c$ satisfies the conditions \ref{framework-1} and \ref{framework-2} in \Cref{thm:framework}.
    
    When $k < \frac{4}{\varepsilon} $, run the algorithm in \Cref{lem: diverse-knapsack-exact} on the parameters in \Cref{eq:f-double-parameter} with $ d_{min} = 0$ since we allow repeated solutions.
    Then, the running time of the algorithm is
    \[\begin{split}
        O\Bigl((n+1) (d_{min}+1)^{\binom{k}{2}} (W+1)^k(P+1)^k2^k k^2 \Bigr) &= O\left( (n+1) \left(\left(2+\frac{1}{\gamma}\right)n\right)^{\frac{4}{\varepsilon}} \left(\left(5+\frac{1}{\delta}\right)n\right)^{\frac{4}{\varepsilon}} 2^{\frac{4}{\varepsilon}}\left(\frac{4}{\varepsilon}\right)^2  \right) \\
        &=(\delta\gamma)^{O(-\varepsilon^{-1})}n^{O(\varepsilon^{-1})}.
    \end{split} \]
    Since the algorithm for the case $k\geq\frac{4}{\varepsilon}$ runs in time $O( n^3 \delta^{-1})k^4 $ (\Cref{lem: diverse-knapsack-beta-k}), the overall running time simplifies to \RunTimeDKDP.
\end{proof}

\section{Algorithms for \dmwisPG and \dmwvcPG (\Cref{thm:bi-apx-pg})}\label{sec:DIS-PG}

In this section, we prove \Cref{thm:bi-apx-pg}.
For the reader's convenience, we restate the theorem here.
\RestBiApxISPG*

Before we proceed, we illustrate \emph{Baker's technique} and define \emph{nice tree decomposition}, both of which are closely related to our technique.

\paragraph{Baker's Technique.} Finding one maximum independent set (MIS) and one minimum vertex cover (MVC) in planar graphs is NP-complete~\cite{garey1979computers}, and an influential work by Baker~\cite{baker1994approximation} provided a PTAS for both problems. 
A planar graph $G = (V, E)$ can be embedded in the plane and the \emph{levels} of vertices can be computed in linear time~\cite{hopcroft1974efficient,lipton1979separator}.
A vertex is at level 1 if it is on the exterior face. In general, the vertices on the exterior face after all vertices at levels up to $i - 1$ have been removed are said to be at level $i$. We refer to the collection of all vertices at a certain level as a \emph{layer}. 

\begin{figure}[h]
    \centering
    \includegraphics[width=.45\textwidth]{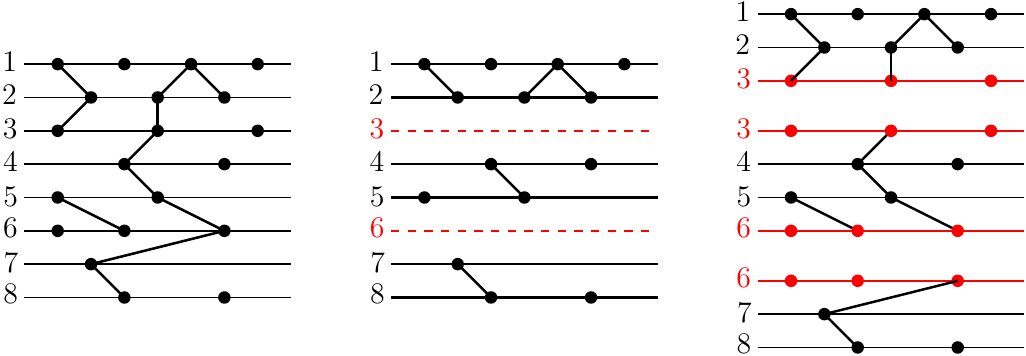}
    \caption{\small{Illustration of decomposition of $G$. Here, $G$ consists of 8 layers and $\ell=2$.
    The left one denotes a part of $G$ with the level indicating to the left of each layer. The middle one is a collection of $\ell$-outerplanar graphs constructed by removing $ H^{(0)}_{3}(G)$ from $G$. The right one is a collection of $(\ell+2)$-outerplanar graphs constructed by duplicating $ L^{(0)}_{3}(G)$.}}
    \label{fig: strata}
\end{figure}

Given an approximation factor $(1-\gamma)$, let $ \ell = (1/\gamma) -1$, and let $ p \in\{0, \ldots, \ell\}$. The \emph{$ p $-th strata of $G$}, denoted $H^{(p)}_{\ell+1}$ (or simply $\pthstrata$ when clear from the context), is the set of all vertices in $ G $ that are at levels congruent to $ p $ modulo $ \ell+1 $, i.e., the collection of every $ (\ell+1) $-st layer from the $ p $-th layer; see \Cref{fig: strata}. A planar graph is said to be $\ell$-outerplanar if it has at most $\ell$ layers. Baker's technique proceeds in two steps:

\begin{enumerate}[leftmargin=*]
    \item Baker shows that there exists an MIS $S'$ and a $p \in \{0,\ldots,\ell\}$ such that removing the $p$th strata from $S'$ results in an independent set $S$ such that $|S| \geq (1-\gamma)|S'|$. That is, ignoring the vertices in the strata does not decrease the size of an MIS by a factor more than $(1-\gamma)$. For MVCs, Baker shows that there exist layers such that duplicating them does not increase the size too much.
    \item Removal of such a $p$th strata now decomposes $G$ into a collection of several $\ell $-outerplanar graphs $G_1,\cdots,G_q$, for $\ell = (1/\gamma) - 1$. Baker then provides a dynamic programming based exact algorithm for independent sets and vertex covers in $\ell$-outerplanar graphs running in time $ 2^{O(\ell)}n$.
    The final solution is the union of all optimal solutions for the $\ell$-outerplanar graphs $G_1,\cdots,G_q$.
    For the MIS problem, by 1 above, the returned solution is at least $(1-\gamma)$ of optimal.
\end{enumerate}

Using the notation in \Cref{sec: framework} with $c=1-\gamma$, if we let $\mathcal{A}_{(1-\gamma)}$ denote Baker's algorithm above, then the solutions output by $\mathcal{A}_{(1-\gamma)}$ belong in a smaller class $\cF^{''}_{1-\gamma} \subseteq \cF_{(1-\gamma)}$ defined as: 
\[
\cF_{1-\gamma}^{''} \coloneqq \left\{
\begin{aligned}
S \subseteq V : &\exists p \in \{0,\ldots,\ell\}\ \text{s.t.}\ S \cap H^p=\emptyset,\\
&\text{and } S \text{ is an MIS in } G[V-H^{(p)}]
\end{aligned}
\right\},
\]
where $G[A]$ denotes the graph induced by some $A\subseteq V$.

In other words, the space of solutions is restricted in that the independent sets returned by Baker's algorithms will be missing vertices from a certain strata. By ignoring a strata we lose some $c$-optimal solutions from the solution space. Thus any generalization of Baker's algorithm is likely to be insufficient for solving DNO problem over $\cF_{c}$.

\paragraph{Nice Tree Decomposition.}
Our algorithm runs on a nice tree decomposition of an $\ell$-outerplanar graph. We first formally define a \emph{tree decomposition} of a graph and define the ``nice'' version of it.

\begin{definition}[Tree decomposition \cite{kleinbergtardos2005}]\label{def: tree-decomp}
A \emph{tree decomposition} (TD) of a graph $ G = (V, E) $ consists of a tree $T$ and a subset $V_t \subseteq V, $ called the \emph{bag}, associated with each node $t$ of $T$, such that the ordered pair $ (T, \{V_t: t\in T\}) $ must satisfy the following three properties:
    \begin{enumerate}[leftmargin=*,label=\normalfont(\alph*)]
        \item Every node of $G$ belongs to at least one bag $V_t$.
        \item For every edge $e$ of $G$, there is some bag $V_t$ containing both ends of $e$.
        \item Let $t_1$, $t_2$ and $t_3$ be three nodes of $T$ such that $t_2$ lies on the path from $t_1$ to $t_3$. Then, if a vertex $v$ of $G$ belongs to both $V_{t_1}$ and $V_{t_3}$, it also belongs to $V_{t_2}$.
\end{enumerate}
\end{definition}
\paragraph{Remark.} To avoid confusion, we use the term \emph{a vertex} for $v\in V$ and \emph{node} for a vertex $t$ of a tree decomposition of $T$ of $G$.
The width of a tree decomposition $T$ is defined as $ \max_{t\in T}(|V_t| - 1) $, and the \emph{treewidth} of a graph $G$ is the minimum width over all tree decompositions of $G$. Unless stated otherwise, $\omega$ denotes the width of a tree decomposition.

\vspace{2mm}Now, we define a nice tree decomposition of a graph.

\begin{definition}[Nice TreeDecomposition, \cite{bodlaender2008combinatorial}] \label{def: NTD} A tree decomposition $T$ is called a ``nice'' tree decomposition if each node $t$ is one of the following four types.
    \leavevmode
    \begin{enumerate}
        \item (Leaf) A node $t$ is leaf of $T$, and $ |V_t| = 1 $.
        \item (Introduce) A node $t$ has exactly one child, say $t'$, and there is a vertex $v\in V$ such that $ V_t = V_{t'} \cup \{v\} $ and $ v \not\in V_{t'} $.
        \item (Forget) A node $t$ has exactly one child, say $t'$, and there is a vertex $ v \in V $ such that $ V_{t'} = V_t \cup \{ v \} $ and $ v \not\in V_t $.
        \item (Join) A node $t$ has exactly two children, say $ t' $ and $ t'' $ and $ V_t = V_{t''} = V_{t''} $.
    \end{enumerate}
\end{definition}

\paragraph{Remark.} Since every $\ell$-outerplanar graph has a treewidth of at most $3\ell-1$~\cite{Bodlaender1988}, using the algorithms in Lemma 7.4 and Theorem 7.18 of~\cite{cygan2015parameterized}, any $\ell$-outerplanar graph can be transformed in time $2^{O(\ell)}n^2$ into a nice tree decomposition with width of $O(\ell)$ and $O(\ell n)$ nodes.
Since these running times do not affect the overall running time of our algorithm, we assume throughout this section that any $\ell$-outerplanar graph can be transformed into a nice tree decomposition without any additional running time.

\vspace{2mm} Now, we are ready to prove \Cref{thm:bi-apx-pg}.

\subsection{Applying \NoCaseChange{\Cref{thm:framework}}}
We first describe our approach in the unweighted setting for the \dmwisPG problem.
We then outline (i) the algorithm for \dmwvcPG and (ii) how to extend this to the weighted setting.
Let $ G=(V,E)$ be a planar graph, and let $ \FInd $ denote the family of all independent sets in $ G $. Fix an integer $ k\geq 1 $ and parameters $ c, \delta, \varepsilon \in (0,1) $. Unless stated otherwise, these objects and parameters remain fixed throughout this section.

First we boost Baker's analysis to show that there exists a strata, called \emph{marginal strata}, such that removing it does not decrease the size of \emph{any} of the maximally-diverse $c$-optimal solutions by too much, \emph{and} the removed vertices do not decrease the diversity of the $c$-optimal solutions by too much.

\begin{restatable}{lemma}{ResMarginalStrata}{\normalfont{[Existence of Marginal Strata]}}
\label{lem:marginal-strata-dis}
    Given a planar graph $G=(V,E)$, $ c, \delta, \varepsilon \in (0,1) $ and an integer $ k \geq 1 $, let $ \ell \geq 4k\delta^{-1} + 2\varepsilon^{-1} -1 $.
    Then, for any subsets $ S_1, \ldots, S_k $ of $V$, there exists some $ p \in \{0, \ldots, \ell\} $ such that the following conditions simultaneously hold:
    \begin{enumerate}[label=(\roman*)]
        \item $ |S_h \cap H^p| \leq \frac{\delta}{4}|S_h| \quad \text{for all $h\in[k]$, and}$
        \item $\sum_{i< j} | (S_i \cap H^p) \triangle (S_j \cap H^p) | \leq \frac{\varepsilon}{2} \sum_{i < j}|S_i \triangle S_j|$.
    \end{enumerate}
\end{restatable}
\begin{proof}
    The proof is a simple application of the union bound. Fix any independent sets $S_1,\ldots,S_k$ of $G$. Let $ \ell \geq 4k\delta^{-1} + 2\varepsilon^{-1} - 1 $.
    Assume for contradiction that every $ p\in\{0, 1, \ldots, \ell \}$ satisfies at least one of the following:
    \begin{enumerate}[label=(\roman*$'$)]
        \item $ |S_h \cap H^p| > \frac{\delta}{4} |S_h| \text{ for some $h\in[k]$.}$ \label{item:neg-quality-marginal-strata}
        \item $ \sum_{i < j} | (S_i \cap H^p) \triangle (S_j \cap H^p) | > \frac{\varepsilon}{2} \sum_{i < j}|S_i \triangle S_j|$.\label{item:neg-diversity-marginal-strata}
    \end{enumerate}
    For each $h\in [k]$, let
    $ A_h = \{ p: |S_h \cap H^p| > \frac{\delta}{4}|S_h| \} $, i.e., the set of all $p$'s satisfying the condition~\ref{item:neg-quality-marginal-strata}.
    We claim that $ |A_h| < 4\delta^{-1} $ for every $h\in[k]$, as, otherwise, 
    \begin{equation*}
    \begin{split}
        |S_h| &= \sum_{ p \in\{0, \ldots, \ell\}} |S_h \cap H^p|
            \geq \sum_{p \in A_h} |S_h \cap H^p| \\
            &> \sum_{p \in A_h} \frac{\delta}{4}  |S_h|
            \geq {4\delta}^{-1} \left( \frac{\delta}{4} \cdot |S_h| \right)
            = |S_h|,
    \end{split}
    \end{equation*}
    which is a contradiction. Similarly, define $|A_{div}|$ as the set of all $p$'s satisfying the condition (ii') above.
    Then, $ |A_{div}| < 2{\varepsilon}^{-1} $, because, otherwise,
    \[
    \begin{split}
        \sum_{i < j} | S_i \triangle S_j |
            &= \sum_{p\in\{0, \ldots, \ell\}} \Big( \sum_{i < j} | (S_i \cap H^p) \triangle (S_j \cap H^p) | \Big) \\
            &\geq \sum_{p\in A_{div}} \Big( \sum_{i < j} | (S_i \cap H^p) \triangle (S_j \cap H^p) | \Big) \\
            &> 2{\varepsilon}^{-1} \left(\frac{\varepsilon}{2} \sum_{i < j} | S_i \triangle S_j |\right) 
            = \sum_{i < j} | S_i \triangle S_j |.
    \end{split}
    \]
    Therefore, it follows that $|A_1|+\cdots+|A_k|+|A_{div}| < 4{k}{\delta}^{-1} + 2{\varepsilon}^{-1} \leq \ell+1,$
    which is a contradiction because every $p$ must belong to either $A_h$ for some $h$ or $A_{div}$, and $|\{0, \ldots, \ell\}| = \ell+1$.
    Hence, there exists $p \in \{0, \ldots, \ell\} $ satisfying conditions (i) and (ii) if $\ell \geq 4k\delta^{-1} + 2\varepsilon^{-1} - 1$.
\end{proof}

\medskip Let $\widehat{S} $ denote a $ (1-\frac{\delta}{4}) $-maximum independent set, which can be computed by Baker's technique in time $ 2^{O(\delta^{-1})}n $.
Let $\widehat{\alpha}$ denote the size of $\widehat{S}$.
We fix $\widehat{S}$ and $\widehat{\alpha}$ throughout this section, unless otherwise stated.

\paragraph{$\cF_{c}^{''}$ in Condition~\ref{framework-1}, Condition~\ref{framework-2} and Condition~\ref{framework-exact}. } \label{sec: condition-4}

Let $ \ell'' \coloneqq \lceil 4k\delta^{-1} + 2\varepsilon^{-1} -1 \rceil $, and let $\pthstrata$ denote the vertices in the marginal strata guaranteed by the preceding \Cref{lem:marginal-strata-dis}, where $ p \in  \{0, 1, \ldots, \ell'' \} $.
Define $\cF''_c$ as
\begin{equation}
    \cF''_c \coloneqq \left\{
    \begin{aligned}\label{eq:pg-f-double-prime}
        S \subseteq \FInd : &S \cap H^p= \emptyset \text{ and $ |S| \geq (1-\frac{\delta}{4})c\widehat{\alpha} $}
    \end{aligned} \right\}.
\end{equation}

\begin{proposition}
    $\cF''_c$ in \Cref{eq:pg-f-double-prime} satisfies Conditions~\ref{framework-1}, \ref{framework-2} and \ref{framework-exact} of \Cref{thm:framework}.
\end{proposition}
\begin{proof}
Verifying Condition \ref{framework-1} is straightforward; since $ (1-\frac{\delta}{4})^2 \geq (1-\frac{\delta}{2}) $, $\cF''_{c}$ is contained in the space of all $(1-\frac{\delta}{2})c$-maximum independent sets.

To verify Condition \ref{framework-2}, we show that $\cF''_c$ contains $S''_1,\ldots,S''_k$ such that $ \sum_{i<j}|S''_i\triangle S''_j| \geq (1-\frac{\varepsilon}{2})\mathrm{OPT_{\mathrm{div}}(\cF_{c})} $.
Let $S_1,\ldots,S_k$ be $c$-maximum independent sets of $G$ with optimal diversity. Note that $(S_i-H^p)$ is an independent set of $G[V-H^p]$ for every $i\in[k]$. By the definition of marginal strata, we have 
\[\begin{split}
    |S_i-H^p|
    &= |S_i|-|S_i\cap H^p|
    >|S_i| - ({\delta}/{4})|S_i| = \left(1-{\delta}/{4}\right)|S_i| \\
    &\geq \left(1-{\delta}/{4}\right)c\alpha(G) \geq \left(1-{\delta}/{4}\right)c\widehat{\alpha}.
\end{split}\]
For diversity, note that $ \sum_{i<j} |(S_i - H^p) \triangle (S_j - H^p)| = \sum_{i<j}((|S_i \triangle S_j| -  |(S_i\cap H^p) \triangle (S_j \cap H^p)|) $.
Thus,
\[
\begin{split}
    \sum_{i<j}|(S_i - H^p) \triangle (S_j - H^p)|
    &\geq \sum_{i<j}|S_i\triangle S_j| - \frac{\varepsilon}{2} \sum_{i<j}|S_i\triangle S_j| \\
    &= \left(1-\frac{\varepsilon}{2}\right)\sum_{i<j}|S_i\triangle S_j|;
\end{split}
\]
hence Condition \ref{framework-2} is satisfied.
\end{proof}

For Condition \ref{framework-exact},
we use the existing result by Baste et al.~\cite{baste2022diversity} that runs in time $f''(n,k,\sigma,c, \delta,\varepsilon)= 2^{O(\ell'' k)} n^{O(k)} $.

\paragraph{ $\cF'_c$ in Condition~\ref{framework-1}, Condition~\ref{framework-2}, and $\cA_{\BCO}$ in Condition~\ref{framework-kbest}.}\label{sec: condition-3}

Note that to obtain a polynomial time algorithm, $ \cA_{\BCO} $ should run in polynomial time.
Any generalization of Baker's algorithm does not run in polynomial time on the same decomposition of $G$ used to construct $\cF''_c$ especially when $ k = \Omega(\frac{1}{\varepsilon})$ since the outerplanarity parameter $\ell$ depends on $k$.
Recall that $  \cA_{\BCO} $ in Condition~\ref{framework-kbest} outputs only a single farthest solution. We exploit this fact to remove the dependence of $\ell$ on $k$ in $ \cA_{\BCO} $.
Let $\widehat{\alpha} = |\widehat{S}| $. We fix $\widehat{S}$ and $\widehat{\alpha}$ throughout this sections.

We begin by constructing $\cF'_c$.

\begin{restatable}{lemma}{ResMarginalBCO}{\normalfont{[Construction of $\cF'_c$]}}\label{lem:mis-f-prime}
    Let $ \ell' \coloneqq \lceil 2\delta^{-1} + 2\varepsilon^{-1} - 1 \rceil $, and for every $p\in\{0,1, \ldots, \ell'\}$, let $ \cF^{\,p} = \{ S \subseteq V-H^p: S \in \FInd \text{ and } |S| \geq (1-\frac{\delta}{2})c\widehat{\alpha} \} $. Define
    \begin{equation}\label{eq:mis-f-prime}
        \cF'_c := \bigcup_{p\in\{ 0,1,\ldots,\ell' \}} \cF^{\,p}.
    \end{equation}
    Then, $\cF'_c$ satisfies Conditions~\ref{framework-1} and \ref{framework-2} of \Cref{thm:framework}.
\end{restatable}
\begin{proof}
    We first show that \Cref{eq:mis-f-prime} satisfies Condition~\ref{framework-1} of \Cref{thm:framework}.
    If $ S \in \cF^{\,p} $, then,
\[ |S| \geq (1-\frac{\delta}{2})c\widehat{\alpha} \geq (1-\frac{\delta}{2})(1-\frac{\delta}{4})c\alpha(G) \geq (1-\delta)c\alpha(G). \]
Therefore, $ \cF'_c \subseteq \cF_{(1-\delta)c} $.

    Now, we show that \Cref{eq:mis-f-prime} satisfies Condition~\ref{framework-2} of \Cref{thm:framework}.
    Fix arbitrary $k$ $c$-maximum independent sets $S_1, S_2, \ldots, S_k$ with optimal diversity, i.e., $S_h\in\cF_c$ for all $h\in[k]$ and $ \sum_{i<j}|S_i \triangle S_j| = \mathrm{OPT}_{\mathrm{div}}(\cF_c)$.
    Let $ A_h =\{ p \in [0,\ell']: |S_h \cap H^{p}| > \frac{\delta}{2}|S_h| \}$ and
    let $ B_h = \{ p \in [0,\ell']: |S_h \cap H^{p}| \leq \frac{\delta}{2}|S_h| \} $ for every $h\in[k]$.
    Then, arguing similarly along the line in \Cref{lem:marginal-strata-dis}, $ |A_h| \leq \frac{2}{\delta}$. Thus, $ |B_h| = \ell'+1 - |A_h| \geq \ell'+1-\frac{2}{\delta} \geq \frac{2}{\varepsilon} $. We claim that
    \[
        \Expect_{p_h \sim B_h, h\in[k] }\left[ \sum_{i<j} \left\lvert \Big(S_i - H^{p_i}\Big) \triangle \Big( S_j - H^{p_j} \Big) \right\rvert \right] \geq (1-\frac{\varepsilon}{2})\sum_{i<j}|S_i \triangle S_j|,
    \]
    where each $p_h$ is drawn from $B_h$ independently and uniformly at random for every $h\in[k]$.
    Fix $ i, j \in[k] $ and $ v \in S_i \triangle S_j $.
    Without loss of generality, assume that $ v \in S_i - S_j $. Then,
    \begin{equation}
    \begin{split}
        \Prob\left\{ v \notin \left( S_i - H^{p_i} \right)\triangle\left( S_j - H^{p_j}\right)\right\}
        =\Prob\left\{ v \in H^{p_i} \right\}
        \leq \frac{1}{|B_i|}
        \leq \frac{\varepsilon}{2}
    \end{split}
    \end{equation}
    Thus, $ \Prob\left\{ v \in \left( S_i - H^{p_i} \right)\triangle\left( S_j - H^{p_j}\right)\right\} > (1-\frac{\varepsilon}{2}) $, and therefore
    \[ \Expect_{p_i\sim B_i, p_j\sim B_j}\left[ \left\lvert \Big(S_i - H^{p_i}\Big) \triangle \Big( S_j - H^{p_j} \Big) \right\rvert \right] > (1-\frac{\varepsilon}{2})\sum_{i<j}|S_i \triangle S_j|. \]
    By linearity of expectation, the claim follows. Hence, there exists $ p'_1, p'_2, \dots, p'_k \in [0,\ell'] $ such that \[ \sum_{i<j}\left\lvert \Big(S_i - H^{p_i}\Big) \triangle \Big( S_j - H^{p_j} \Big) \right\rvert > ( 1-\frac{\varepsilon}{2}) \sum_{i<j}|S_i \triangle S_j| \]
    and this verifies Condition~\ref{framework-2} of \Cref{thm:framework}. This completes the proof.  
\end{proof}

Recall that $ \cF'_c $ is the union of $ \cF^{\,p} $s. Therefore, any $\cA_{\BCO}(G,|\cdot|,r,\cF^{\,p})$ can be easily extended to $\cA_{\BCO}(G,|\cdot|,r,\cF'_c)$ by running $\cA_{\BCO}(G,|\cdot|,r,\cF^{\,p})$ for all $p\in\{0,1, \ldots, \ell'\}$ and taking the best output among them.
This extension increases the running time by a factor of $\ell'$, which is merely $O(n)$.

Let $G_1$ and $G_2$ be $\ell$-outerplanar graphs generated by removing strata $H^p$ from $G$.
Let $ S^1_1, S^1_2, \ldots, S^1_k$ and $ S^2_1, S^2_2, \ldots, S^2_k$, respectively, be maximally diverse collections of maximum independent sets of $G_1$ and $G_2$. Then,
\[ \begin{split}
    \sum_{i<j}| (S_i-H^p) \triangle (S_j-H^p) | &= \sum_{i<j} | (S_i \cap V(G_1))\triangle (S_j \cap V(G_1))| + \sum_{i<j} | (S_i \cap V(G_2))\triangle (S_j \cap V(G_2))| \\
    &\leq \sum_{i<j}|S^1_i \triangle S^1_j| + \sum_{i<j}|S^2_i \triangle S^2_j|
    =\sum_{i<j}| (S^1_i \cup S^2_i)\triangle(S^1_j \cup S^2_j) |,
\end{split} \]
where the last inequality holds since any solution of $S^1_1, \ldots S^1_k$ is disjoint from any solution from $S^2_1, \ldots, S^2_k$.
This means that we may compute the solutions from each $\ell$-outerplanar graphs and merge them in a natural way.

Hence, for an $ \cA_{\BCO}(G,|\cdot|,r,\cF'_c) $, it suffices to design an algorithm, running in a single nice tree decomposition (or a single $\ell$-outerplanar graph), that maximizes $r( S )$ subject to $ |S| \geq B $ for some prescribed budget $B\geq0$.

\begin{restatable}{lemma}{MisBCO}\label{lem:mis-bco}
    Let $T$ be a tree decomposition with width $O(\ell')$ and $O(\ell' n)$ nodes of a graph $ G =(V, E) $.
    Let $ r: V \rightarrow \mathbb{Z} $ be an objective function such that $ r(S) = \sum_{v\in S}r(v) $ and $ r(S) \in [-R,R] $ for every $S\subseteq V$ for some $R\geq0$.
    Then, given $B\geq0$, there exists a $ 2^{O(\ell)}R^2n $-time algorithm that returns $ \argmax \{ r(S): \text{$S$ is an independent set of $G$ and 
    } |S| \geq B \} $.
\end{restatable}
\begin{proof}
    Let $T$ be a nice tree decomposition of $G$, with width $O(\ell')$ and $O(\ell'n)$ nodes. 
    For any node $t$ of $T$, let $G^t$ denote the subgraph of $ G $, induced by the subtree of $T$, rooted at node $t$.
    Define $  f(t, U, R') $ as the size of a maximum independent set $S^t$ of $G^t$ such that $ S^t \cap V_t = U $ and $ r(S^t) = R' $ for each $ h\in[k] $.
    In each case below, let $ U $ denote any fixed subset of $ V_t $ and $R$ denote any fixed nonnegative integer.
    Assume further that for each node $t$ of $T$, for all subsets $U$ of $V_t$ and for all $ R' \in [-R, R] $, set $ f(t,U,R') = -\infty $ initially, so that we can avoid manually handling the unsatisfying cases.

    \vspace{2mm}\textbf{(Case 1: $t$ is a leaf node.)}\quad
    Let $t$ be a leaf node, and let $V_t = \{v\}$. Then, $V_t$ is the only independent set of $G^t$. Therefore, $ f(t, \{v\}, r(v)) = 1 $ and $ f(t, \emptyset, 0) = 0 $.

    \vspace{2mm}\textbf{(Case 2: $t$ is an introduce node.)}\quad
    If $t$ is an introduce node, then $t$ has exactly one child, say $t'$, and there exists a unique vertex $ v \in V_t $ introduced at node $t$, \emph{i.e.}, $ V_t = V_{t'} \cup \{v\} $ and $ v \not\in V_{t'}$.
    Note that $G^t$ is obtained from $G^{t'}$ by adding the vertex $v$ and edges connecting $v$ to vertices in $V_{t'}$. In particular, $v$ is not adjacent to any vertex in $V(G^{t'})$ but $V_{t'}$.
    Therefore, if $U$ is an independent set of $G[V_t]$ and $U \ni v$, then $ f(t, U, R') = f(t', U-\{v\}, R'-r(v)) + 1 $.
    If $U$ is an independent set of $G[V_t]$ but $ v \notin U $, then, $ f(t,U,R') = f(t',U,R') $.

    \vspace{2mm}\textbf{(Case 3: $t$ is a forget node)}
    Let $t$, $t'$ and $v$ be the same as above except that $v$ is forgotten at node $t$.
    Let $U$ be an independent set of $G[V_t]$. Note that $ G^t $ and $ G^{t'} $ are the same graph.
    Therefore, the maximum weight of an independent set of $G^{t}$ that does not contain $v$ can be computed as $ f(t', U, R')$, and that containing $v$ can be computed as $ f(t', U \cup \{v\}, R')$.
    Whichever the greater one is the maximum weight of an independent set of $G^t$.
    
    \vspace{2mm}\textbf{(Case 4: $t$ is a join node)} Let node $t$ be a join node. Then, $t$ has exactly two children, say $ {t_1}$ and ${t_2}$ and $ V_t = V_{t_1} = V_{t_2} $.
    Note that $ G^{t_1} $ and $ G^{t_2} $ are separated by $ V_t $ and $ G^t = G^{t_1} \cup G^{t_2} $~\cite{bodlaender2008combinatorial}. Moreover, for an independent set $S^t$ of $G^t$, there are two corresponding independent sets $ S^{t_1} \subseteq G^{t_1} $ and $ S^{t_2} \subseteq G^{t_2}$ such that $ S^t = S^{t_1} \cup S^{t_2} $.
    In particular, if $ S^t \cap V_t = U $, then $ S^{t_1} \cap V_t = S^{t_2} \cap V_t = U $ and $ S^{t_1} \cap S^{t_2} = U $, thus, $ |S^t| = |S^{t_1}| + |S^{t_2}| - |U| $ and $ r(S^t) = r(S^{t_1}) + r(S^{t_2}) - r(U) $.
    Therefore, $ f(t, U, R') = \max
        \{ f(t_1, U, R'_1) + f(t_2, U, R'_2) - |U| : R'_1+R'_2 = R'+r(U) \} $.
    Since $R'_1,R'_2 \in [-R,R]$, $f(t,U,R')$ can be computed in time $ O(R^2) $ for a fixed $t$ and $U$.

    \vspace{2mm}\textbf{(Analysis and Construction of Actual solutions.)}
    At the root node, choose the largest $R'$ among all $U$ such that $ f(\mathrm{root}, U, R') \geq B $.
    
    Since the computation time for a join node is greater than that of any other type of node, the total running time is $ |V(T)| $ times the computation time for a join node.
    Recall the bag associated with each node can contain $O(\ell')$ vertices of $ G $.
    Therefore, the time required for a join node is $ 2^{O(\ell')} \cdot R^2 $.
    Since $ |V(T)| = O(\ell'n)$, the overall running time is $ 2^{O(\ell')}\, R^2 \, n $.
    
    To retrieve the actual solutions, we may create a separate table for each node and record which independent set of the node is taken. Then, without sacrificing the overall running time, we may backtrack the execution and obtained the actual solutions.
\end{proof}

\paragraph{Putting everything together.} With all conditions verified, an application of \Cref{thm:framework} proves \Cref{thm:bi-apx-pg}. The runtime is $ \max( 2^{O(\ell')} n^2k^2 n (k^2 \log k),   2^{O(\ell'' /\varepsilon)} n^{O(4/\varepsilon)})$, where $\ell' = \lceil 2\delta^{-1}+2\varepsilon^{-1} - 1 \rceil $ and $\ell'' \leq 2(4/\varepsilon)\delta^{-1}+2\varepsilon^{-1} - 1 $, which simplifies to \divmisruntime.

\medskip\medskip Now, we illustrate the algorithm for \dmwvcPG and how to extend our approach to the weighted setting.

\paragraph{Diverse $c$-Minimum Vertex Covers.}

We briefly outline the algorithm for the \textsc{Diverse $c$-Minimum Vertex Covers} problem.

Recall that, for each $p \in \{0,1,\ldots,\ell\}$, we duplicate the vertices in the strata $H^p$. This decomposes $G$ into a disjoint collection of $(\ell+2)$-outerplanar graphs.
The remaining procedure is similar to that for the \textsc{Diverse $c$-Maximum Independent Sets} problem.
However, the independent-set algorithm cannot be applied verbatim.
When the duplicated copies are identified in order to project the solutions back to $G$, diversity contributed by the duplicated layers may disappear.

Let $\widetilde G_p := G[V \oplus H^p]$ denote the graph obtained by duplicating the vertices of $H^p$. We color all copies of vertices in $H^p$ red and all remaining vertices black.
For a vertex cover $S$ of $\widetilde G_p$, let $\operatorname{blk}_p(S)$ denote the set of black vertices contained in $S$.

When $k=O(1/\varepsilon)$, we modify the bounded-treewidth dynamic program so that it computes vertex covers $S_1,S_2,\ldots,S_k$ of $\widetilde G_p$ satisfying
\[  |S_i| \leq \widehat \beta \qquad\text{for every } i\in[k],\]
where $\widehat \beta$ is an appropriate quality threshold, while maximizing
\begin{equation}\label{eq:mvc-div-small-k}
    \sum_{i<j}| \operatorname{blk}_p(S_i)\triangle\operatorname{blk}_p(S_j)|
\end{equation}
Thus, red vertices contribute to the quality constraint but not to the diversity objective.
Each $S_i$ is projected back to a vertex cover of $G$ by identifying the duplicated copies.
Since this projection does not change the membership of black vertices, the diversity of the projected solutions is at least their black-vertex diversity. By the marginal strata lemma (\Cref{lem:marginal-strata-dis}), there exists some $p\in\{0,1,\ldots,\ell\}$ for which the diversity in \Cref{eq:mvc-div-small-k} is at least a $(1-O(\varepsilon))$-fraction of the maximum diversity achievable from the target space $\cF_c$.

\medskip The preceding construction requires all $k$ solutions to use common marginal strata.
A direct application of the marginal strata lemma to all $k$ benchmark solutions requires
$\ell=\Omega(k/\delta+1/\varepsilon)$, which is unsuitable when $ k= \Omega(\frac{1}{\varepsilon})$.
We relax this dependence on $k$ using the following observation.

\begin{observation}
Let $\mathcal{S}^*=\{S_1^*,S_2^*,\ldots,S_k^*\}$ be a collection of maximally diverse $c$-minimum vertex covers, and let $ D^* :=\sum_{i<j}|S_i^*\triangle S_j^*| $.
For every integer $t\in[k]$, there exist, not necessarily distinct,
solutions $R_1,R_2,\ldots,R_t\in\mathcal{S}^*$ and indices
$a_1,a_2,\ldots,a_k\in[t]$ such that, after setting
$T_i:=R_{a_i}$ for every $i\in[k]$,
\[
    \sum_{i<j}|T_i\triangle T_j|
    \geq
    \left(1-\frac{1}{t}\right)
    \left(1-\frac{1}{k}\right)D^*.
\]
\end{observation}

\begin{proof}
Choose $R_1,R_2,\ldots,R_t$ independently and uniformly at random from $\mathcal{S}^*$, and independently choose
$a_1,a_2,\ldots,a_k$ uniformly at random from $[t]$. For every $i<j$, we have
\[ \Expect[ |T_i\triangle T_j| ] = \frac{1}{t^2} \sum_{a,b\in[t]} \Expect[|R_a\triangle R_b|] = \left(1-\frac{1}{t}\right)\frac{2D^*}{k^2}.
\]
Consequently,
\[
\mathbb{E}\!\left[\sum_{i<j}|T_i\triangle T_j|\right]
    =\binom{k}{2}\left(1-\frac{1}{t}\right)\frac{2D^*}{k^2}
    =\left(1-\frac{1}{t}\right) \left(1-\frac{1}{k}\right)D^*.
\]
Hence, a deterministic choice satisfying the claimed bound exists.
\end{proof}

We apply the marginal-strata argument to the
$t$ representative solutions.
The same counting argument shows that, for $ \ell=\Theta(\frac{t}{\delta}+\frac{1}{\varepsilon}) $ there exists marginal strata $H^p$ such that duplicating $H^p$ increases the weight of every representative by at most an $ O(\delta) $-fraction of the weights and removes at most an $O(\varepsilon)$-fraction of the diversity of the representatives.
Note that by setting $t=\Theta(1/\varepsilon)$, we have
$\left(1-\frac{1}{t}\right)\left(1-\frac{1}{k}\right)= 1-O(\varepsilon)$ whenever $k=\Omega(\frac{1}{\varepsilon})$, and $ \ell = O(\frac{1}{\delta \varepsilon} + \frac{1}{\varepsilon} )$ which is independent of $k$.

Although the larger value of $\ell$ makes the large-$k$ routine slightly slower than its independent-set counterpart, where $\ell' = O(\frac{1}{\delta}+\frac{1}{\varepsilon})$, the overall running time remains the same due to the small-$k$ regime.

\medskip\paragraph{Extending to the Weighted Setting via Scaling and Rounding.}
We now consider a vertex-weighted planar graph with weight function
$w\colon V\to\mathbb{R}_{\geq 0}$. We extend the preceding algorithms using a scaling-and-rounding scheme analogous to \Cref{lem: knapsack-scaling}. For the maximum-weight independent-set
problem, let $S$ be a sufficiently accurate approximate solution and,
for an internal error parameter $\delta'=\Theta(\delta)$, define
\[
    \widetilde W
    :=
    \left\lceil\frac{1-\delta'}{\delta'}\,n\right\rceil,
    \qquad
    \widetilde w(v)
    :=
    \left\lfloor
        \frac{\widetilde W+n}{c\,w(S)}\,w(v)
    \right\rfloor.
\]
We then search over independent sets satisfying $\widetilde w(I)\geq\widetilde W$.
For minimum-weight vertex cover, we
analogously use upward rounding together with an upper rounded-weight
threshold.
\section{Diverse Enclosing-Polygons}\label{sec: diverse-enclosing-polygons}

In this section, we provide our result for the \textsc{Diverse Enclosing-Polygons}.
The \textsc{Diverse Enclosing-Polygons} problem is a diverse version of the \textsc{Fence Enclosure} problem~\cite{arkin1993geometric}, where the input is a set of points $P=\{p_1,\dots,p_n\}$, with an integer value $v_i $ associated to $p_i$, and a budget $L \in \mathbb{R}_{>0}$. The goal of the \textsc{Fence Enclosure} problem is to find a polygon of perimeter at most $L$ that encloses\footnote{Enclosing refers to {\it weakly enclosing}, i.e., the boundary points on edges also are included.} a set of points of maximum total value.
Now, we introduce a diverse version of the \textsc{Fence Enclosure} problem.

\begin{tcolorbox}[colback=white,colframe=black,sharp corners,boxrule=0.5pt]
    \noindent\textsc{Diverse Enclosing-Polygons}
    \vspace{4mm}
    
    \noindent\textbf{Input:} A set $P = \{p_i \in \mathbb{R}^2 : i \in [n]\}$ of $n$ points, where each point $p_i$ for $i \in [n]$ is associated with an integer profit $v_i \geq 0$. Additionally, three parameters are specified: numbers $ c \in (0, 1) $, $L > 0$ and an integer $k \geq 1$. 

    \vspace{2mm}
    \noindent\textbf{Output:} $k$ distinct subsets $P_1,\ldots,P_k$ of $P$ that satisfy the following conditions simultaneously, if the exist.
    Let $\mathcal{S}_L$ denote the collection of all the subsets of $P$ whose convex hulls have perimeters at most $L$. For any $S$ in $\mathcal{S}_L$, let $v(S)$ be $\sum_{p_i \in S, i \in [n]} v_i$.
    
    \begin{enumerate}[leftmargin=*]
        \item For each $i \in [k]$, $P_i \in \mathcal{S}_L$.
        \item For each $i \in [k]$, $v(P_i) \ge c\cdot\max_{S \in \mathcal{S}_L} v(S)$.
        \item $ \sum_{i \ne j \in [k]} |P_i \Delta P_{j}| $ is maximized.
    \end{enumerate}
\end{tcolorbox}

As a partial application to our framework (\Cref{cor:framework}), we prove that there is a $\tbetak$-approximation algorithm for $(1-\delta)c$-resource augmentation algorithm for the \textsc{Diverse Enclosing-Polygons} problem, by adapting the dynamic programming by Arkin et al.~\cite{arkin1993geometric}.

\begin{theorem}[Diverse Enclosing-Polygons]
    \label{thm:enclosingpolygon}
    \textsc{Diverse Enclosing-Polygons} admits an $\enclPolyRun$-time $\tbetak$-approximation algorithm with $(1-\delta)c$-resource augmentation.
\end{theorem}
\begin{proof}

    \begin{figure}[h]
        \centering
        \includegraphics[width=0.35\textwidth]{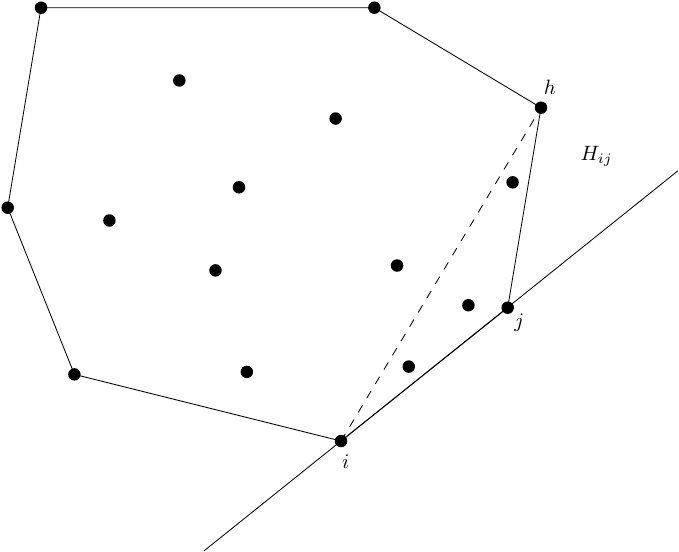}
        \caption{An illustration of an enclosing polygon}
        \label{fig:enclosing_polygon}
    \end{figure}

    We prove that there is a $ O(\delta^{-1} n^5 k^2) $-time algorithm for the BCBE problem in the space $\cF'_{(1-\delta)c}$. By \Cref{cor:framework}, then the result follows.
    First, by following the scaling and rounding scheme of \Cref{lem: knapsack-scaling} one can find a $(1-\delta)$-optimal solution by using the algorithm in \cite{arkin1993geometric}. Call it $S$.
    Now, according to \Cref{lem: knapsack-scaling}, scale and round the value $v_i$ of each point to $ \tilde{v} = \left\lfloor \frac{\tilde{V}+n}{c\cdot v(S)}\cdot v_i \right\rfloor $, where $ \tilde{V} = \left\lceil \frac{1-\delta}{\delta} \cdot n \right\rceil $.
    Our goal then reduces to solve the BCBE problem with budget $\tilde{V}$.
    As usual, define the score $s_i$ of the point $p_i$ as $s_i = \sum_{m=1}^{k}\mathbbm{1}(p_i\not\in P_m) - \sum_{m=1}^{k}\mathbbm{1}(p_i\in P_m)$.
    Then, the budget-constrained version of the problem asks to find a point set $P'$ such that $ \sum_{p_i\in P'}s(p_i) $ is maximized, $P'$ encloses value $ \tilde{V} $ and its length is at most $L$, where the length of a point set denotes the length of its convex hull.
    
    Let $\ell(s, p_i,p_j,\tilde{v})$ be the minimum length of the enclosing polygon to enclose the value exactly $\tilde{v}$ and the total score of the enclosed points is exactly $s$, subject to the constraints that
        \begin{enumerate}
            \item the enclosure lies within the halfplane $H_{ij}$ (left to the oriented line $\overrightarrow{ij}$), and
            \item $\overline{p_ip_j}$ is an edge of the enclosing polygon.
        \end{enumerate}
        Then, $\ell(s, p_i, p_j, \tilde{v}) $ is defined recursively by
    \begin{equation*}
        \min\limits_{\substack{h:\, p_h \in H_{ij},\\ h \neq i,j}} \left\{ \ell\Big(\big(s-s(\triangle_{ijh})+s_i+s_h\big), p_i, p_h, \big(\tilde{v} - \tilde{v}(\triangle_{ijh}) +\tilde{v}_i+\tilde{v}_h\big)\Big) + (\ell_{ij} + \ell_{jh} - \ell_{ih}) \right\},
    \end{equation*}
    where $\tilde{v}(\triangle_{ijh})$ (resp., $s(\triangle_{ijh})$ represents the sum of the (adjusted) values (resp., the scores) of all the points in $P$ enclosed by the triangle formed by the three points $p_i$, $p_j$ and $p_h$, and $\ell_{ij}$ is the length of the line segment connecting $p_i$ and $p_j$, and similarly for $\ell_{jh}$  and $\ell_{hi}$; see \Cref{fig:enclosing_polygon} for illustration. 
    It is possible to answer the triangle queries in constant time~\cite{arkin1993geometric}.
    The base cases for the above recurrence relation are: 
    \begin{enumerate} 
        \item $\ell(s, p_i, p_j, \tilde{v})=\infty$ if $\tilde{v} \neq \sum_{h:\, p_h \in H_{ij}} \tilde{v}_h$ or $s \neq \sum_{h:\, p_h \in H_{ij}} s_h$ 
        \item $2l_{ij}$ if $\tilde{v}_i+\tilde{v}_j=\tilde{v}$ and $s_i+s_j=s$. 
    \end{enumerate}
    We compute $\min_{i\neq j}\ell(s, p_i, p_j, \tilde{V})$ in order of increasing $\tilde{v}$ and $s$, for $ \tilde{v} =1, \ldots, \tilde{V} $ and for $s = 1, \ldots, nk$, and find $ \min_{i\neq j}\ell(s, p_i, p_j, \tilde{V}) < \infty $ with the largest possible value of $s$. 
    If $ \min_{i\neq j}\ell(s, p_i, p_j, \tilde{V}) = \infty $ for all $ s $, then report $\infty$. 
    Since $s_{\max} \leq nk $ and $\tilde{V}=O(\tfrac{n}{\delta}$, the running time of this step is $O(s_{\max}n^3\tilde{V}) $.  $ O(\delta^{-1} n^5 k) $.

    Akin to the diverse knapsack problem, this algorithm can be turned into a $ O(\delta^{-1} n^5 k^2) $-time $k$-best enumeration procedure (see the proof of \Cref{lem: diverse-knapsack-beta-k}), and this verifies Condition \ref{framework-kbest}.
    
    We can now apply \Cref{cor:framework} to obtain the desired $\betak$-approximation algorithm with runtime $O(\delta^{-1} n^5 k^4 \log k)$, which simplifies to $ O(\enclPolyRun) $.
\end{proof}
\section{Diverse TSP Tours}\label{sec:diverse-tsp}

In this section, we present our results for the \textsc{Diverse TSP} problem. We begin with a precise problem definition.

\begin{tcolorbox}[colback=white,colframe=black,sharp corners,boxrule=0.5pt]

    \noindent\textsc{Diverse TSP}

    \vspace{2mm}
    \noindent\textbf{Input:} A complete graph $G=(V,E)$ with $n$ vertices, $\ell_{ij}$ for every $(i,j)\in E$, and an integer $k\geq1$, a niceness factor $c$.
    
    \noindent\textbf{Output:} $k$ distinct $c$-optimal TSP tours $T_1,\ldots,T_k$ that maximize $ \sum_{i < j} \lvert T_i \Delta T_j \rvert $.
\end{tcolorbox}

The works in~\cite{do2020evolving,do2022analysis} studied related problems using an evolutionary algorithm. The algorithm in~\cite{do2020evolving} works as follows: Given a factor $\delta$, the algorithm starts with a set $P = \{T_1, \ldots, T_k\}$ of $k$ $\tfrac{1}{1+\delta}$-optimal TSP tours. Each tour $T_i$ in this set satisfies $\ell(T_i) \leq (1+\delta)\,T_{\mathrm{opt}}$, where $\ell$ denotes the total length of the tour. 
At each step, a randomly selected tour $T_i \in P$ is “mutated” to produce a new tour $T_i'$. If $T_i'$ remains $\tfrac{1}{1+\delta}$-optimal, it is added to $P$, and one of the $k+1$ tours is removed to maximize the diversity of $P$ at that step. This iterative process continues until a predefined termination criterion is met. However, the algorithm does not provide a guaranteed worst-case running time.

In the subsequent work~\cite{do2022analysis}, the authors analyze the algorithm for small values of $k$. However, this analysis focuses on permutations rather than tours, meaning the quality of the output tours is not explicitly considered.
In contrast, we present two algorithms for the \textsc{Diverse TSP} problem, each with a guaranteed running time. The first is a partial application to our framework (\Cref{cor:framework}), and is a $\betak$-approximation algorithm running in $O(2^n \, n^4 \, k^4 \log^2 (k))$. The second algorithm finds a pair of TSP tours that are farthest apart in time $O(4^n \, n^5)$.

\begin{restatable}{theorem}{ResDiverseTSP}{\normalfont{[Diverse TSP]}}\label{thm: diverse-tsp} For the \textsc{Diverse TSP} problem, the following hold:
    \begin{enumerate}[leftmargin=*,label=\normalfont(\arabic*)]
        \item There is an $ O^*(k^4 \, 2^n) $-time $\betak$-approximate algorithm for the \textsc{Diverse TSP} problem.
        \item There is an $ O^*(4^n) $-time algorithm that finds two optimal TSP tours $T_1$ and $T_2$ that maximize $\lvert T_1\Delta T_2\rvert$.
    \end{enumerate}
\end{restatable}

\vspace{2mm} Bellman~\cite{bellman1962dynamic} and Held and Karp~\cite{held1962dynamic} used dynamic programming (BHK-DP, henceforth) to compute a single TSP tour in time $O(n^2 \cdot 2^n)$. The idea of BHK-DP is to check every combination of remaining cities before the current city.
More precisely, if $L(i, S)$ denotes the optimal TSP tour length that starts from vertex 1 and ends at vertex $ i \in S \subseteq V\setminus\{1\} $, $ L(i, S) $ can be computed as the minimum of $ L(j, S \setminus \{i\}) + \ell_{ij} $ over all $ j \in S \setminus \{1, i\} $. We adapt the idea of this dynamic programming to our framework.

\begin{proof}[Proof of \Cref{thm: diverse-tsp}]
Assume that an instance of the \textsc{Diverse TSP} problem has been given with $G=(V,E)$, $\ell_{ij}$, a positive integer $k$, and a niceness factor $c$. Additionally, $T^*$ be an optimal TSP tour. This can be found in time $O(n^2\,2^n)$ by running the BHK-DP algorithm.

Given $k$ $c$-optimal TSP tours $T_1,\ldots,T_k$, for each edge $ e $ of $E$, define $$ w(e) = \sum_{h\in[k] E} (\mathbbm{1}(e\not\in T_h) - \mathbbm{1} (e\not\in T_h) ).$$
Let $L_k(W',i,S) = (t'_1,\ldots,t'_k)$ denote the tuple of the lengths of the $k$-best $c$-optimal TSP tours $ T'_1,\ldots,T'_k $ each of which has weight $W'$ and starts at vertex $1$ and ends at $i \in S \subseteq V\setminus\{1\}$, i.e., $ \ell(T'_1) \leq \cdots \leq \ell(T'_k) \leq \ell(T) $ for any $c$-optimal TSP tour $T$ with weight $W'$ that has starts at vertex 1 and ends at vertex $ i \in S \subseteq V\setminus\{1\} $.
If there are less than $k$ such tours, the rest of the components of $L_k(W',i,S)$ is filled with $\infty$.
Then, $ L_k(W',i,S)$ can be computed by choosing the best $k$ weights from the following set:
\begin{equation}\label{eq:tsp-far-recur}
    \left\{L\Big(W'-w(i,j),j,S\setminus \{i\} \Big) + \ell_{ij} :j \in S \setminus \{1,i\} \right\}.
\end{equation}
Note that the $k$ smallest weights in increasing order among $O(n)$ objects can be found in time $O(n+k\log(k))$.
Since each edge can be contained in each tour at most once, $ W' $ is at most $nk$.
Therefore, the entire dynamic programming table can be filled in time $ O( (n^3k + n^2k^2\log(k))\, 2^n ) $, which is simply $ O(n^3\,k^2\log(k)\,2^n) $.
Once the entire table has been filled, start scanning from $ W'=nk $ down to $W'=0$ while selecting $k$ weights with tour length at most $ \tfrac{1}{c}\,\ell(T^*) $.
By our framework in \Cref{thm:framework}, the overall running time for the $\betak$-approximate algorithm for the \textsc{Diverse TSP} problem is $ O( n^4\,k^4\log^2(k)\,2^n ) $, thus the desired time bound in (1) follows.

Let $ L(C, i_1, i_2, S_1, S_2) $ be the sum of the lengths of two optimal TSP tours, say $T_1$ and $T_2$, such that both $T_1$ and $T_2$ start at node 1, tour $T_1$ ends at node $i_1 \in S_1 \subseteq V \setminus \{1\} $, tour $T_2$ ends at node $i_2 \in S_2 \subseteq V \setminus \{1\} $, and $|T_1 \cap T_2|=C$. Then,
    \begin{align*}
        &L(C, i_1, i_2, S_1, S_2)\\
        &= \displaystyle\min_{\substack{j_1 \in S_i \setminus \{1,i_1\}\\j_2 \in S_2 \setminus \{1,i_2\}}} L\left( C - f(i_1,j_1,i_2,j_2), j_1, j_2, S_1 \setminus\{i_1\}, S_2 \setminus \{i_2\} \right) + g(i_1,j_1,i_2,j_2),
    \end{align*}
where $ f(i_1,j_1,i_2,j_2) = 1 $ and $ g(i_1,j_1,i_2,j_2) = 2\ell_{i_1j_1} $ if $ i_1j_1 = i_2j_2 $, and otherwise, $ f(i_1,j_1,i_2,j_2) = 0 $ and $ g(i_1,j_1,i_2,j_2) = \ell_{i_1j_1} + \ell_{i_2j_2} $.
Note that the bases cases are $L(C, i_1, i_1, \{1,i_1\}, \{1,i_1\}) = 2\ell_{1i_1}$ if $C=1$; $L(C, i_1, i_2, \{1,i_1\}, \{1,i_2\}) = \ell_{1i_1} + \ell_{1i_2}$ if $i_1 \neq i_2$ and $C=0$. All other cases are error cases, where $L(C, i_1, i_2, S_1, S_2) = \infty$.

Starting from $C=0$ to $C=n-1$, find the minimum value of $C$ such that $L(C,1,1,V,V) = 2\ell(T^*)$.
Since every tour has $n-1$ edges and the returned $C$ denotes the minimum number of common edges of the two tours, the diameter of the optimal TSP tours of $G$ is $ 2(n-1)-C $.
Since $ C < n$, the overall running time of this algorithm is $ O( 4^n \cdot n^5 )$; thus, the desired time bound in (2) follows.
\end{proof}

\section{Conclusion and Future Work}

In this work, we provide a general framework for designing bi-criteria approximation algorithms that compute diverse solutions to packing and covering problems. As applications of this framework, we obtain bi-PTASs for the diverse variants of several classical problems, including \KP, MWIS and MWVC. We also establish complementary lower bounds, showing that \DK and \dmwisPG do not admit a standard PTAS or a bi-criteria FPTAS.

A gap remains between our upper and lower bounds. A natural direction for future work is therefore to determine whether these problems admit bi-criteria EPTASs (efficient polynomial-time approximation schemes).

Another promising direction is to investigate whether our framework extends to other packing and covering problems. One interesting candidate is the \textsc{Minimum Dominating Set} (MDS) problem on unit-disk graphs. Since this problem admits a PTAS~\cite{hunt1998nc}, techniques similar to those used in our framework for planar graphs may also be applicable in this setting. Given the fundamental role of MDS on unit-disk graphs in theory and practice, we believe that this direction merits further investigation.

\clearpage
\bibliographystyle{plainurl}
\bibliography{reduced-ref}


\end{document}